\documentclass{article} 
\usepackage[final]{colm2026_conference}

\usepackage{microtype}
\usepackage{hyperref}
\usepackage{url}
\usepackage{booktabs}

\definecolor{darkblue}{rgb}{0, 0, 0.5}
\hypersetup{colorlinks=true, citecolor=darkblue, linkcolor=darkblue, urlcolor=darkblue}

\usepackage{graphicx}
\usepackage{breqn}
\usepackage{multirow}
\usepackage{array}
\usepackage{subcaption}
\usepackage{algorithm}
\usepackage{algpseudocode}
\usepackage{enumitem}
\usepackage{nicefrac}
\usepackage{amsfonts}
\usepackage{amsmath}
\usepackage{amsthm}
\usepackage{tikz}
\usepackage{pifont}
\usepackage{wrapfig}
\newcommand{\cmark}{{\textcolor{green!70!black}{\ding{51}}}}
\newcommand{\xmark}{{\textcolor{red}{\ding{55}}}}

\newtheorem{theorem}{Theorem}
\newtheorem{lemma}{Lemma}

\usepackage{xcolor}

\title{OPERA: Online Data Pruning for Efficient Retrieval Model Adaptation}

\author{\textbf{Haoyang Fang, Shuai Zhang, Yifei Ma, Hengyi Wang,} \\
  \textbf{Cuixiong Hu, Katrin Kirchhoff, Bernie Wang, George Karypis} \\
  Amazon Web Services \\
  \texttt{\{haoyfang,shuaizs,yifeim,yuyawang\}@amazon.com}}

\let\cite\citep

\begin{document}

\maketitle

\begin{abstract}

Domain-specific finetuning is essential for dense retrievers, yet not all data pairs contribute equally to the learning process. We introduce OPERA\footnote{Code is released at \url{https://github.com/autogluon/autogluon-rag/tree/main/projects/opera}.}, a data pruning framework that exploits this heterogeneity to improve both the effectiveness and efficiency of retrieval model adaptation. We first investigate static pruning (SP), which retains only high-similarity query-document pairs, revealing an intrinsic quality-coverage tradeoff: ranking (NDCG) improves while retrieval (Recall) can degrade due to reduced query diversity. To resolve this tradeoff, we propose a two-stage dynamic pruning (DP) strategy that adaptively modulates sampling probabilities at both query and document levels throughout training, prioritizing high-quality examples while maintaining access to the full training set. Evaluations across eight datasets spanning six domains demonstrate the effectiveness of both approaches: SP improves ranking over standard finetuning (NDCG@10 +0.2 points), while DP achieves the strongest performance on both ranking (NDCG@10 +1.0 points) and retrieval (Recall@20 +0.4 points), with an average rank of 1.38 across all methods. These findings scale to Qwen3-Embedding, an LLM-based dense retriever, confirming architecture-agnostic benefits. Notably, DP reaches comparable performance in less than 50\% of the training time required by standard finetuning.

\end{abstract}

\section{Introduction}

Dense retrievers have advanced information retrieval~\cite{contriever, sbert, dpr, bge, gritlm, e5, nvemb}, substantially outperforming traditional sparse methods~\cite{bm25, tfidf}. Built on pretrained language models~\cite{transformer, bert, mistral}, these models achieve strong zero-shot performance across diverse benchmarks~\cite{beir, mteb}. Nevertheless, achieving optimal performance on specific downstream tasks still requires domain-specific finetuning~\cite{dadense, ulmfit}.

Data pruning and coreset selection have shown promises for improving training efficiency in neural networks~\cite{el2n, lessismore, beyond, glister}, while dynamic pruning methods have particularly demonstrated that adjusting data selection during training can maintain performance while reducing computation~\cite{infobatch, dpdiffusion, dpsample}. However, these methods are designed for standard classification or generation tasks where training samples are treated as independent, identically distributed instances. Dense retriever finetuning is fundamentally different: it employs a two-stage contrastive sampling framework~\cite{bge, bgem3, bgellm} where queries are first sampled, and then positive and negative documents are selected for each query. This hierarchical structure means that data quality operates at two distinct granularities (query relevance and document relevance), creating unique challenges that existing pruning methods do not address. To our knowledge, no prior work has studied data pruning specifically for dense retriever finetuning. A detailed discussion of related work is provided in Appendix~\ref{appendix:related_work}.

\begin{figure*}[t]
    \centering
    \begin{tikzpicture}
        \node[inner sep=0] (left) {\includegraphics[width=0.7\textwidth]{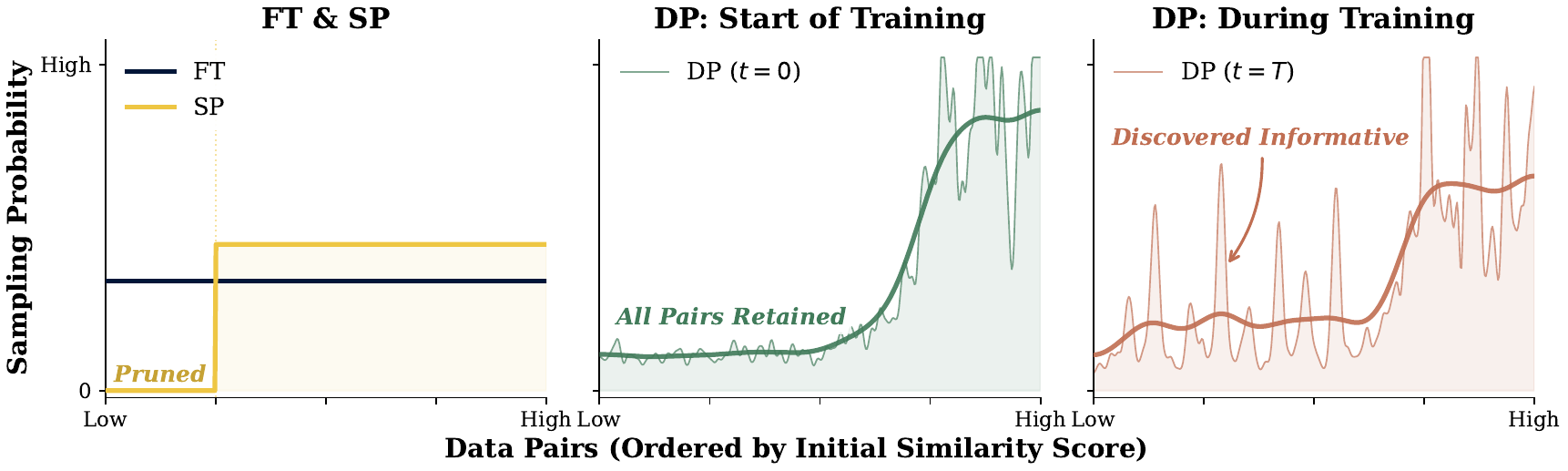}};
        \node[inner sep=0, anchor=west] (right) at ([xshift=0.3cm]left.east) {\includegraphics[width=0.27\textwidth]{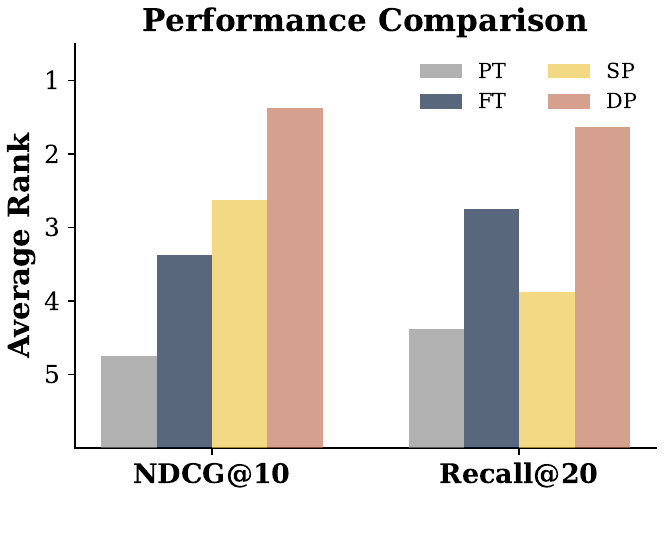}};
        \draw[->, thick, black!66, line width=1.2pt]
            ([yshift=-2pt]left.east) -- ([yshift=-2pt]right.west);
    \end{tikzpicture}
    \caption{\textbf{Comparison of sampling probability distributions across training strategies.} \textit{Left three panels:} Standard finetuning (FT) samples all data pairs uniformly, while static pruning (SP) discards the lowest-similarity ones and up-weights the rest, improving ranking but reducing query coverage. Dynamic pruning (DP) at the start of training retains all pairs with non-zero probability, assigning higher sampling weights to high-similarity pairs through soft thresholding. As training progresses, DP redistributes sampling weights: the model discovers informative examples in the low/mid similarity range and up-weights them. \textit{Right panel:} Average rank of each method on NDCG@10 and Recall@20 across 8 datasets (lower rank is better); DP achieves the best rank on both metrics. See Section~\ref{subsec:sampling_vis} for empirical sampling weight evolution on real data.}
    \label{fig:teaser}
\end{figure*}

We introduce OPERA, a framework that exploits the heterogeneous quality of training data~\cite{d4, semdedup} to improve domain adaptation for dense retrievers (Figure~\ref{fig:teaser}). Our investigation begins with static pruning (SP), which retains only the highest-similarity query-document pairs for training. This simple strategy reveals a key insight: quality-based filtering consistently improves ranking metrics (NDCG) but can degrade retrieval coverage (Recall), because pruning disproportionately removes queries with fewer high-quality documents, breaking the balanced sampling that retrievers rely on for broad coverage. This quality-coverage tradeoff is intrinsic to the two-stage sampling structure of retrieval training and motivates the need for a more nuanced approach.

To resolve this tradeoff, we propose dynamic pruning (DP), which maintains the complete training set while adaptively adjusting sampling probabilities at both query and document levels. Unlike InfoBatch~\cite{infobatch}, which uses a fixed loss-average threshold and rescales gradients to maintain unbiasedness, our approach implements dynamic thresholds that evolve throughout training via cosine scheduling, and preserves original learning rates to emphasize high-quality training signals. Rather than employing hard exclusions, DP assigns reduced but nonzero sampling probabilities to lower-quality examples, ensuring continued data diversity while prioritizing informative instances.

Concretely, our dynamic pruning framework combines three components: (1) hierarchical pruning at both query and document granularities, reflecting the two-stage sampling structure; (2) dynamic threshold scheduling that progressively sharpens selection as model representations improve; and (3) soft pruning mechanisms that modulate sampling probabilities rather than discarding data.

\vspace{-1pt}
We evaluate OPERA on eight datasets spanning nutrition~\cite{nfcorpus, beir}, medicine~\cite{tripclick}, finance~\cite{fiqa, beir}, non-factoid QA~\cite{antique}, factoid QA~\cite{triviaqa, dpr, hotpotqa, beir}, and fact verification~\cite{fever, beir}, including both datasets seen and unseen during pretraining. Our contributions are:
\vspace{-1pt}
\begin{itemize}
    \item We identify a quality-coverage tradeoff unique to retrieval's two-stage sampling: static quality filtering improves ranking but degrades recall. This finding holds across encoder-only (BGE) and LLM-based (Qwen3-Embedding) retrievers.

    \item We propose dynamic pruning with hierarchical scheduling that resolves this tradeoff, improving both ranking and recall while halving convergence time. We provide an efficient formulation compatible with fixed-iteration training frameworks.

    \item We validate OPERA across 8 datasets, 6 domains, and 2 architectures (encoder-only BGE and decoder-based Qwen3-Embedding). We further demonstrate effective denoising through a combined SP+DP strategy, and provide theoretical guarantees on when pruning outperforms standard finetuning.
\end{itemize}

\section{Methodology}

\subsection{Preliminary: Standard Finetuning (FT)}
Given a dataset $(Q,D)$ with $n$ queries where each query $q \in Q$ has $m_q$ positive documents from the document corpus $D$, we adopt the standard two-stage contrastive sampling framework for dense retriever finetuning~\citep{bge}, in which queries are sampled first and positive/negative documents are then sampled per query. For each training step $t$, we first uniformly sample queries from $Q$. Then, for each sampled query, we randomly select one positive document $d$ from its $m_q$ positive documents and one negative document $d'$ through hard negative mining (HNM) for contrastive loss computation~\citep{bge}. Our pruning framework operates on this two-stage structure, modulating the query- and document-level sampling probabilities rather than the contrastive objective itself.

The sampling probabilities are:
\begin{equation}
P_t(q) = \nicefrac{1}{n}, \quad P_t(d | q) = \nicefrac{1}{m_q}, \quad P_t(d'|q) = HNM(D,q)
\end{equation}

For notational simplicity, we omit the time step $t$ and $P(d'|q)$ for the following equations.

\subsection{Static Pruning (SP)}
\label{subsec:sp}

To understand how data quality affects retrieval finetuning, we begin with a straightforward approach inspired by~\citet{beyond}: retaining only the highest-quality training pairs. We compute cosine similarity between each query-document pair using the pretrained model and retain only the top fraction of pairs for training.
We define $D(q)$ as the set of positive documents for query $q$, and use indicator variables $I_q$ and $I_{q,d}$ to denote whether a query or query-document pair is kept by the SP algorithm:
\begin{equation}
I_q, I_{q,d} =
\begin{cases}
1, & \text{if kept} \\
0, & \text{if pruned}
\end{cases}
\end{equation}
This adjusts our sampling probabilities to:
\begin{equation}
P(q) = \frac{\sum_{d \in D(q)} I_{q,d}}{\sum_{q \in Q} \sum_{d \in D(q)} I_{q,d}}, \quad P(d|q) = \frac{I_{q,d}}{\sum_{d \in D(q)} I_{q,d}},
\end{equation}
where $Q$ is the set of all queries.
In practice, we select top $k\sum_{q \in Q} m_q$ query-document pairs with highest similarity scores for training, where $k \in [0,1]$ is the data retention rate:
\begin{equation}
k = \frac{\sum_{q \in Q} \sum_{d \in D(q)} I_{q,d}}{\sum_{q \in Q} m_q}
\end{equation}
Given the retention rate $k$, the indicators are set deterministically: we rank all query-document pairs by similarity and set $I_{q,d}=1$ for the top-$k$ fraction and $I_{q,d}=0$ otherwise, and set $I_q=1$ if and only if at least one document of query $q$ is retained (i.e., $\sum_{d \in D(q)} I_{q,d} > 0$). A query is thus dropped entirely when all of its positive pairs fall below the cutoff.

\vspace{-4pt}

Crucially, quality-based filtering breaks the uniform query sampling that standard finetuning relies on: queries with more high-similarity documents are overrepresented, while queries with fewer such documents may be excluded entirely. This improves ranking (e.g., NDCG) by focusing on well-matched pairs, but can degrade recall by reducing coverage of the query space. This quality-coverage tradeoff is intrinsic to the two-stage sampling structure of retrieval training and motivates the dynamic approach described in Section~\ref{subsec:dp}.

\textbf{Alternative scoring metric.}
We also evaluate consistency-based scoring (CBS)~\cite{e5}, which ranks each positive pair against random negatives and retains pairs that consistently rank highly, adapted to our finetuning setting. CBS and cosine similarity yield comparable results (Appendix~\ref{appendix:cbs}), but CBS is incompatible with DP as it requires re-embedding all negative documents after each model update. We therefore adopt cosine similarity as the default metric for both SP and DP.

\subsection{Dynamic Pruning (DP)}
\label{subsec:dp}
\begin{algorithm}[h!]
\caption{Dynamic Pruning}
\label{alg:dp}
\begin{minipage}[t]{0.46\textwidth}
\begin{algorithmic}
\Procedure{OPERA}{}
\State $\alpha_s \gets \text{query strength start}$
\State $\alpha_e \gets \text{query strength end}$
\State $r_s \gets \text{query ratio start}$
\State $\beta_s \gets \text{doc strength start}$
\State $\beta_e \gets \text{doc strength end}$
\State $v_s \gets \text{doc ratio start}$
\State $v_e \gets \text{doc ratio end}$
\State $n \gets \text{size of dataset}$
\State $t_{max} \gets \text{max training steps}$
\State $n_0 = \lfloor n \cdot (1 - r_s) / \alpha_s + r_s \cdot n \rfloor$
\For{each iteration}
\State $t \gets \text{current training step}$
\State Update $\mathbf{q}$ \Comment{Query scores}
\State Update $\mathbf{p}$ \Comment{Pos.\ scores}
\State $\textsc{SampleQuery}()$
\State $\textsc{SampleDocument}()$
\EndFor
\EndProcedure
\end{algorithmic}
\end{minipage}%
\hfill
\begin{minipage}[t]{0.52\textwidth}
\begin{algorithmic}
\Procedure{SampleQuery}{}
\State $\alpha \gets \alpha_e + \frac{(1 + \cos(\frac{t}{t_{max}} \pi)) \cdot (\alpha_s - \alpha_e)}{2}$
\State $r \gets \frac{\alpha \cdot n_0 - n}{(\alpha - 1) \cdot n}$
\State $q_{top} \gets \text{Top $r$ Queries}$
\State $q_{rem} \gets queries \setminus q_{top}$
\State $q_{rand} \gets \text{Random Select $n_0 - r$ from $q_{rem}$}$
\State $q' \gets q_{top} \cup q_{rand}$
\State \text{Sample from $q'$}
\EndProcedure
\Procedure{SampleDocument}{}
\State $\beta \gets \beta_e + \frac{(1 + \cos(\frac{t}{t_{max}} \pi)) \cdot (\beta_s - \beta_e)}{2}$
\State $v \gets v_e + \frac{(1 + \cos(\frac{t}{t_{max}} \pi)) \cdot (v_s - v_e)}{2}$
\State Calculate the threshold $T$ at cutoff $v$
\State $w = (\mathbf{p} > T) \cdot (\beta - 1) + 1$
\State $w = normalize(w)$
\State \text{Sample documents with weights $w$}
\EndProcedure
\end{algorithmic}
\end{minipage}
\end{algorithm}

The quality-coverage tradeoff identified in SP arises because hard pruning permanently discards data, reducing the diversity of training signals. Dynamic Pruning (DP) resolves this tradeoff by replacing hard exclusions with soft sampling modulation: high-quality examples receive elevated sampling probabilities, while lower-quality examples remain accessible with reduced frequencies. This preserves the broad data coverage needed for recall while concentrating training effort on informative instances.

As detailed in Algorithm~\ref{alg:dp}, DP uses two scores updated from the model's forward pass at each training step: the positive score $\mathbf{p}_i = s(e_q, e_{d_i})$, the cosine similarity between query and positive document embeddings, and the query score $\mathbf{q}_j$, the mean contrastive loss across query $j$'s positive pairs (offset to ensure positivity). These drive a two-level sampling process: at the query level, top-scoring queries are combined with randomly selected low-scoring queries to maintain diversity; at the document level, sampling weights are assigned based on positive scores while ensuring all documents retain non-zero selection probability. Concretely, in \textsc{SampleDocument}, $\beta$ is the current emphasis strength and $v$ is a quality-threshold percentile that broadens over training (from $v_s{=}0.25$ to $v_e{=}0.5$); we compute the score threshold $T$ at cutoff $v$ and assign each document a weight $w = (\mathbf{p} > T)\cdot(\beta-1)+1$, so that above-threshold (high-quality) documents receive weight $\beta$ and the rest receive weight $1$. After normalization, high-quality documents are sampled $\beta$ times more often than the remaining ones, yet no document is ever assigned zero probability, preserving coverage.

The sampling probabilities for queries and query-document pairs are defined as:
\begin{equation}
P(q) =
\begin{cases}
\frac{1}{(1+\alpha)n - n_0}, & \text{if low quality} \\
\frac{\alpha}{(1+\alpha)n - n_0}, & \text{if high quality}
\end{cases}
\quad,\quad
P(d|q) =
\begin{cases}
\frac{1}{[1-(1-\beta)r_q]m_q} & \text{if low quality} \\
\frac{\beta}{[1-(1-\beta)r_q]m_q}, & \text{if high quality}
\end{cases}
\end{equation}
Here, $n_0 = \lfloor n \cdot (1 - r_s) / \alpha_s + r_s \cdot n \rfloor$ represents a fixed virtual dataset size, ensuring compatibility with training frameworks that require a static dataset size. It is chosen so that the expected number of queries sampled per epoch stays fixed regardless of the sampling strength $\alpha$. $r_q$ is the fraction of high-quality positive pairs for each query. $\alpha$ and $\beta$ are sampling strength parameters that vary during training:
\begin{equation}
\alpha(t) = \alpha_e + \frac{(1 + \cos(\frac{t}{t_{max}} \pi)) \cdot (\alpha_s - \alpha_e)}{2}
\end{equation}
where $\alpha_s$ and $\alpha_e$ denote the start and end sampling strengths respectively, and $t_{max}$ is the maximum training steps. Similarly, $\beta$ follows the same cosine decay schedule. Within a virtual pool of size $n_0$, the $r$ top-scoring queries are always included ($P(\text{top in pool})=1$) while the remaining $n_0 - rn$ slots are filled by uniform selection from the $(1-r)n$ low-scoring queries ($P(\text{remaining in pool}) = (n_0 - rn)/((1-r)n)$). Setting the ratio of these inclusion probabilities equal to the sampling strength $\alpha$ and solving yields $r = (\alpha \cdot n_0 - n)/((\alpha-1)\cdot n)$.

\textbf{Intuition behind the schedule.}
The schedule reflects how much the model's own quality scores can be trusted over the course of training. Early on, near initialization, these scores are unreliable, so DP uses weak selectivity ($\alpha_s{=}2$) and samples broadly and near-uniformly; aggressive selection here would merely amplify noise. As training progresses and the scoring signal improves, DP increases selectivity ($\alpha_e{=}5$) to concentrate on high-quality examples. This is the opposite of standard curriculum learning, which starts with easy examples and gradually broadens: DP instead starts broad and progressively \emph{sharpens} its focus toward high-quality pairs.

\textbf{Update Interval.}
To mitigate the computational overhead associated with frequent score updates and pruning operations, we introduce an update interval parameter $I_u$. We update the query scores $\mathbf{q}$ every $I_u$ iterations, while maintaining the per-iteration random selection of $d$ due to its negligible computational cost. Our empirical analysis demonstrates that this optimization reduces additional computation time from 4.5\% to 1.64\% over the baseline, without impacting performance. The effects of varying $I_u$ are examined in Section~\ref{section:overhead}.

\textbf{Comparison with Adaptive Training Methods.}
\label{subsec:adaptive_training}
Table~\ref{tab:adaptive_comparison} compares DP with related adaptive training paradigms (see Appendix~\ref{appendix:related_work} for a broader discussion). Curriculum learning~\cite{cl}, SPL~\cite{clsp}, and active learning all rely on scores computed before or between training stages rather than \textit{continuous scoring} from the current model weights, and permanently exclude samples rather than maintaining \textit{full coverage}. Additionally, curriculum learning follows a predetermined ordering that is not \textit{adaptive} to model state, while SPL and active learning require expensive scoring or labeling passes. InfoBatch addresses these limitations, but applies uniform sample-level pruning without \textit{hierarchical} sampling or \textit{scheduled} strength. DP performs \textit{lightweight} \textit{adaptive} online update with \textit{continuous scoring} and \textit{full data coverage}, uses a \textit{scheduled} sampling strength, and operates at two \textit{hierarchical} levels reflecting the query-document structure of retrieval training.

\begin{table}[h]
\vspace{-5pt}
\centering
\scalebox{0.78}{
\begin{tabular}{@{}lcccccc@{}}
\toprule
 & \textbf{Lightweight} & \textbf{Adaptive} & \textbf{Continuous Scoring} & \textbf{Full Coverage} & \textbf{Hierarchical} & \textbf{Scheduled} \\
\midrule
Curriculum & \cmark & \xmark & \xmark & \xmark & \xmark & \xmark \\
SPL & \xmark & \cmark & \xmark & \xmark & \xmark & \xmark \\
Active Learning & \xmark & \cmark & \xmark & \xmark & \xmark & \xmark \\
InfoBatch & \cmark & \cmark & \cmark & \cmark & \xmark & \xmark \\
\textbf{DP (ours)} & \cmark & \cmark & \cmark & \cmark & \cmark & \cmark \\
\bottomrule
\end{tabular}
}
\caption{Comparison of adaptive training paradigms.}
\label{tab:adaptive_comparison}
\vspace{-5pt}
\end{table}

\textbf{When Does Pruning Help?}
We provide a formal analysis of the conditions under which data pruning outperforms standard finetuning (Appendix~\ref{appendix:proof}). The key result (Theorem~\ref{theorem:denoise}) shows that any scoring function identifying true positives at a rate higher than the noise rate will improve learned query representations. Under equal selection quality, SP outperforms DP because it completely excludes noisy samples rather than merely down-weighting them. However, when signal quality is low, SP's hard removal reduces the training pool while retaining noise, increasing per-sample exposure to mislabeled examples and the risk of overfitting. DP avoids this by keeping all samples accessible with non-zero probability and increasing its sampling strength $\alpha$ over training as the model improves, becoming more selective when the signal is more reliable. This may explain why DP outperforms SP when initial signal is not robust. This motivates a two-stage approach: applying SP first to discard strong noise, followed by DP to refine. The empirical validation is presented in Section~\ref{section:denoising}.

\section{Experiments}

To evaluate our proposed OPERA approach, we design a series of experiments addressing six key research questions:

\begin{itemize}
    \item \textbf{RQ1:} How does OPERA compare to FT and other pruning methods across domains?

    \item \textbf{RQ2:} Can OPERA's findings scale to LLM-based dense retrievers?

    \item \textbf{RQ3:} How effective is OPERA in handling noisy training data?

    \item \textbf{RQ4:} How does OPERA affect convergence speed and training efficiency?

    \item \textbf{RQ5:} What is the computational overhead of OPERA, and how can it be optimized?

    \item \textbf{RQ6:} How does DP allocate training focus across examples over time?
\end{itemize}

We evaluate our methods on eight datasets spanning six domains: NFCorpus~\cite{nfcorpus} (nutrition), TripClick head/torso~\cite{tripclick} (medical), FiQA~\cite{fiqa} (finance), ANTIQUE~\cite{antique} (non-factoid QA), TriviaQA~\cite{triviaqa} and HotpotQA~\cite{hotpotqa} (factoid QA), and FEVER~\cite{fever} (fact verification). FEVER and HotpotQA were seen during bge-large-en-v1.5 pretraining~\cite{bge}, while the others are unseen. See Appendix~\ref{appendix:dataset} for more datasets statistics.

\subsection{Implementation Details}
\label{sec:app_imp}

We used the bge-large-en-v1.5 model~\cite{bge} (335M parameters) as our primary dense retriever. All hyperparameters were selected to optimize the FT baseline performance and used without any additional tuning for all other methods, ensuring a fair comparison. We also compare against Random Pruning (RP), which retains the same fraction of data as SP but selects pairs randomly; RP consistently underperforms all other methods and is therefore excluded from the main results table but included in the efficiency analysis (Figure~\ref{fig:efficiency}) and the full results in Appendix~\ref{appendix:bge_detailed}. Dataset-specific hyperparameters are detailed in Appendix~\ref{appendix:hyperparams}. For LLM-based retriever experiments (Section~\ref{sec:qwen3}), we use Qwen3-Embedding-0.6B with hyperparameters similarly optimized for its FT baseline; implementation details are provided in Appendix~\ref{appendix:qwen3_impl}.

\subsection{Main Results}

\begin{table}[ht]
\centering
\caption{OPERA vs.\ baselines on bge-large-en-v1.5. Best in \textbf{bold}, second-best \underline{underlined}. $^\dagger$Datasets seen during pretraining.}
\label{tab:results}
\resizebox{\textwidth}{!}{
\begin{tabular}{@{}ll ccccc @{\hspace{8pt}} ccccc@{}}
\toprule
& & \multicolumn{5}{c}{\textit{NDCG@10}} & \multicolumn{5}{c}{\textit{Recall@20}} \\
\cmidrule(lr){3-7} \cmidrule(lr){8-12}
& & & & & \multicolumn{2}{c}{\textbf{OPERA}} & & & & \multicolumn{2}{c}{\textbf{OPERA}} \\
\cmidrule(lr){6-7} \cmidrule(lr){11-12}
\textbf{Domain} & \textbf{Dataset} & \textbf{PT} & \textbf{FT} & \textbf{IB} & \textbf{SP} & \textbf{DP} & \textbf{PT} & \textbf{FT} & \textbf{IB} & \textbf{SP} & \textbf{DP} \\
\midrule
Nutrition & NFCorpus & 0.451 & 0.466 & 0.470 & \textbf{0.491} & \underline{0.480} & 0.232 & 0.300 & \underline{0.304} & 0.267 & \textbf{0.304} \\
\multirow{2}{*}{Medical} & TripClick (h) & 0.219 & \underline{0.298} & 0.295 & 0.270 & \textbf{0.309} & 0.189 & 0.242 & \underline{0.244} & 0.219 & \textbf{0.252} \\
& TripClick (t) & 0.208 & 0.245 & \underline{0.248} & 0.236 & \textbf{0.249} & 0.334 & \underline{0.398} & 0.396 & 0.373 & \textbf{0.400} \\
Finance & FiQA & 0.489 & 0.514 & 0.514 & \underline{0.516} & \textbf{0.524} & 0.602 & \underline{0.639} & \textbf{0.640} & 0.630 & 0.639 \\
Non-Factoid QA & ANTIQUE & 0.543 & \underline{0.575} & 0.572 & 0.564 & \textbf{0.590} & \underline{0.414} & 0.405 & 0.398 & \textbf{0.428} & 0.413 \\
\multirow{2}{*}{Factoid QA} & TriviaQA & 0.487 & 0.481 & 0.482 & \textbf{0.501} & \underline{0.491} & 0.429 & 0.458 & \textbf{0.460} & 0.445 & \underline{0.460} \\
& HotpotQA$^\dagger$ & 0.790 & 0.806 & \underline{0.806} & 0.803 & \textbf{0.812} & 0.807 & \underline{0.836} & 0.835 & 0.781 & \textbf{0.838} \\
Fact Verif. & FEVER$^\dagger$ & 0.868 & 0.892 & 0.893 & \textbf{0.915} & \underline{0.902} & 0.950 & 0.960 & \underline{0.961} & 0.950 & \textbf{0.962} \\
\midrule
\multicolumn{2}{l}{Average} & 0.507 & 0.535 & 0.535 & \underline{0.537} & \textbf{0.545} & 0.495 & \underline{0.530} & 0.530 & 0.512 & \textbf{0.534} \\
\midrule
\multicolumn{2}{l}{Avg. Rank (Unseen)} & 4.67 & 3.33 & 3.00 & \underline{2.67} & \textbf{1.33} & 4.50 & 2.83 & \underline{2.33} & 3.50 & \textbf{1.83} \\
\multicolumn{2}{l}{Avg. Rank (Seen)} & 5.00 & 3.50 & \underline{2.50} & \underline{2.50} & \textbf{1.50} & 4.00 & \underline{2.50} & \underline{2.50} & 5.00 & \textbf{1.00} \\
\multicolumn{2}{l}{Avg. Rank (Overall)} & 4.75 & 3.38 & 2.88 & \underline{2.63} & \textbf{1.38} & 4.38 & 2.75 & \underline{2.38} & 3.88 & \textbf{1.63} \\
\bottomrule
\end{tabular}
}
\vspace{-2pt}
{\footnotesize PT: Pretrained, IB: InfoBatch, SP: Static Pruning, DP: Dynamic Pruning. $^\dagger$ Datasets seen in pretraining.}
\end{table}

\begin{table}[t]
\centering
\caption{Comparison on Qwen3-Embedding-0.6B. Large-corpus datasets are excluded due to the prohibitive cost of embedding millions of documents with an LLM. Best in \textbf{bold}, second-best \underline{underlined}.}
\label{tab:qwen3_results}
\resizebox{\textwidth}{!}{
\begin{tabular}{@{}ll ccccc @{\hspace{8pt}} ccccc@{}}
\toprule
& & \multicolumn{5}{c}{\textit{NDCG@10}} & \multicolumn{5}{c}{\textit{Recall@20}} \\
\cmidrule(lr){3-7} \cmidrule(lr){8-12}
& & & & & \multicolumn{2}{c}{\textbf{OPERA}} & & & & \multicolumn{2}{c}{\textbf{OPERA}} \\
\cmidrule(lr){6-7} \cmidrule(lr){11-12}
\textbf{Domain} & \textbf{Dataset} & \textbf{PT} & \textbf{FT} & \textbf{IB} & \textbf{SP} & \textbf{DP} & \textbf{PT} & \textbf{FT} & \textbf{IB} & \textbf{SP} & \textbf{DP} \\
\midrule
Nutrition & NFCorpus & 0.441 & 0.479 & 0.478 & \textbf{0.487} & \underline{0.479} & 0.211 & \textbf{0.314} & 0.308 & 0.265 & \underline{0.311} \\
Non-Factoid QA & ANTIQUE & 0.518 & 0.496 & 0.506 & \textbf{0.540} & \underline{0.520} & \textbf{0.398} & 0.340 & 0.329 & \underline{0.396} & 0.353 \\
Factoid QA & TriviaQA & 0.467 & 0.489 & 0.484 & \underline{0.493} & \textbf{0.504} & 0.403 & \textbf{0.462} & 0.453 & 0.421 & \underline{0.461} \\
\midrule
\multicolumn{2}{l}{Average} & 0.475 & 0.488 & 0.489 & \textbf{0.507} & \underline{0.501} & 0.337 & \underline{0.372} & 0.363 & 0.361 & \textbf{0.375} \\
\bottomrule
\end{tabular}
}
\end{table}

\vspace{-1pt}
\subsubsection{bge-large-en-v1.5}

We conduct a comparative analysis of SP and DP against the pretrained model~\cite{bge}, standard finetuning (FT)~\cite{bge}, and InfoBatch~\cite{infobatch}. The evaluation metrics are NDCG~\cite{ndcg, beir, mteb} and Recall~\cite{hardneg, bgem3}, assessed at the top 10 and 20 retrievals, respectively, with all methods trained for an equivalent number of iterations.

\textbf{SP} confirms the quality-coverage tradeoff in the previous section: it outperforms FT and InfoBatch in NDCG@10 across unseen, seen, and all datasets on average, achieving the best NDCG@10 on NFCorpus (0.491), TriviaQA (0.501), and FEVER (0.915). As expected from our analysis, this comes at the cost of recall, as SP's Recall@20 average rank drops to 3.88, because \textit{a priori} removal of training pairs reduces query diversity. Despite this, SP offers substantial data efficiency: it drops 75\% of pairs yet improves ranking, and subsequent analysis demonstrates even faster convergence and effective denoising capabilities.

\textbf{DP} resolves the quality-coverage tradeoff, achieving the best performance on both metrics. It achieves the highest average rank on NDCG@10 (1.38) and Recall@20 (1.63), consistently outperforming other methods for both unseen (NDCG@10: 1.33, Recall@20: 1.83) and seen (NDCG@10: 1.50, Recall@20: 1.00) datasets. DP achieves the highest NDCG@10 on 6 of 8 datasets and the highest Recall@20 on 5 of 8. By replacing hard exclusions, DP maintains the broad data coverage needed for recall while concentrating training effort on informative instances. An ablation of DP's hierarchical design is presented in Section~\ref{subsec:ablation_dp}, with additional analysis on the static pruning data retention rate in Appendix~\ref{appendix:sp_rate}.

\subsubsection{Qwen3-Embedding-0.6B}
\label{sec:qwen3}

To investigate whether OPERA generalizes beyond encoder-only models, we evaluate on Qwen3-Embedding-0.6B~\cite{qwen3embedding}, a decoder-based LLM embedding model that employs last-token pooling and instruction-based query encoding, representing a fundamentally different architecture from the CLS-pooling encoder model used above. Due to the higher computational cost and potential data leakage from large-scale pretraining, we use a higher learning rate (1e-5 vs.\ 1e-6) with only 2,000 iterations (vs.\ 8,000--32,000 for BGE). Note that this setting is inherently less favorable to DP, which benefits from more iterations to dynamically adjust sampling rates. We evaluate on datasets with fewer than 1M documents (NFCorpus, ANTIQUE), with the exception of TriviaQA (21M documents), included to demonstrate scalability. We exclude FiQA, as the pretrained Qwen3-Embedding model already outperforms all finetuned methods on this dataset (NDCG@10: 0.511, Recall@20: 0.633), likely due to high-quality financial domain data in its pretraining corpus. As with the BGE experiments, all hyperparameters were first optimized for the vanilla FT baseline, and OPERA's pruning methods were applied without additional tuning.

Table~\ref{tab:qwen3_results} presents the results. Despite the limited training budget, SP achieves the best average NDCG@10 (0.507), while DP achieves the best average Recall@20 (0.375), reproducing the same quality-coverage pattern observed with bge-large-en-v1.5: SP excels at ranking due to its focus on high-quality examples, while DP maintains stronger recall through soft pruning. Notably, even under conditions unfavorable to dynamic pruning (few iterations, high learning rate), DP still outperforms baselines on average, achieving NDCG gains similar to SP while preserving and improving recall, confirming that OPERA's benefits extend to LLM-based retrievers. Detailed results are provided in Appendix~\ref{appendix:qwen3_detailed}.

\subsection{Ablation: Hierarchical Query-Document Pruning}
\label{subsec:ablation_dp}

\begin{table}[t]
\begin{minipage}[t]{0.38\textwidth}
\centering
\caption{Ablation of hierarchical pruning. Recall@20 is 0.639 for all; Recall@10 shown to differentiate.}
\label{tab:ablation_dp}
\resizebox{\textwidth}{!}{
\begin{tabular}{lcc}
\toprule
\textbf{Method} & \textbf{NDCG@10} & \textbf{Recall@10} \\
\midrule
FT & 0.514 & 0.547 \\
DP w/ Query Sel. & 0.517 & 0.556 \\
DP w/ Doc Sel. & 0.516 & 0.549 \\
DP w/ Both & \textbf{0.524} & \textbf{0.559} \\
\bottomrule
\end{tabular}
}
\end{minipage}%
\hfill
\begin{minipage}[t]{0.59\textwidth}
\centering
\caption{Computational overhead of DP with varying query update intervals ($I_u$) on FiQA.}
\label{tab:computation}
\resizebox{\textwidth}{!}{
\begin{tabular}{lccccc}
\toprule
\textbf{Method} & $\mathbf{I_u}$ & \textbf{Iters/sec} & \textbf{Speed Diff (\%)} & \textbf{NDCG@10} & \textbf{Recall@20} \\
\midrule
FT & -- & 2.43 & 0.00 & 0.514 & 0.639 \\
InfoBatch & -- & 2.42 & -0.30 & 0.514 & 0.640 \\
\midrule
DP & 1 & 2.32 & -4.49 & \textbf{0.524} & 0.639 \\
DP & 10 & 2.37 & -2.51 & 0.519 & \textbf{0.646} \\
DP & 100 & 2.39 & -1.64 & 0.522 & 0.635 \\
\bottomrule
\end{tabular}
}
\end{minipage}
\end{table}

A key design choice in DP is operating at two granularities, query selection and document selection, reflecting the two-stage sampling structure of retrieval training. To isolate the contribution of each granularity, we ablate them on FiQA~\cite{fiqa} (Table~\ref{tab:ablation_dp}), comparing three variants.
\textbf{Query selection only} applies the dynamic schedule to $P(q)$ (as in \textsc{SampleQuery}) with uniform document sampling $P(d|q)=1/m_q$. \textbf{Document selection only} keeps query sampling uniform ($P(q)=1/n$) and applies the dynamic schedule to $P(d|q)$ (as in \textsc{SampleDocument}). \textbf{Combined (full DP)} applies both jointly.

Query selection alone (NDCG@10: 0.517, Recall@10: 0.556) and document selection alone (NDCG@10: 0.516, Recall@10: 0.549) both improve over FT (NDCG@10: 0.514, Recall@10: 0.547), confirming that both granularities carry complementary signal, and their combination achieves the highest scores on both metrics. This supports our core argument that retrieval-specific data pruning must account for the query-document hierarchy.

This ablation isolates the \emph{hierarchy}; DP additionally relies on cosine \emph{scheduling} and \emph{soft} sampling, whose individual contributions we isolate through controlled single-component ablations on FEVER in Appendix~\ref{appendix:ablation_components}. An additional ablation on the static pruning data retention rate is provided in Appendix~\ref{appendix:sp_rate}.

\subsection{Efficiency Analysis}

\subsubsection{Convergence Speed}

\begin{figure*}[t!]
    \centering
    \includegraphics[width=0.95\textwidth]{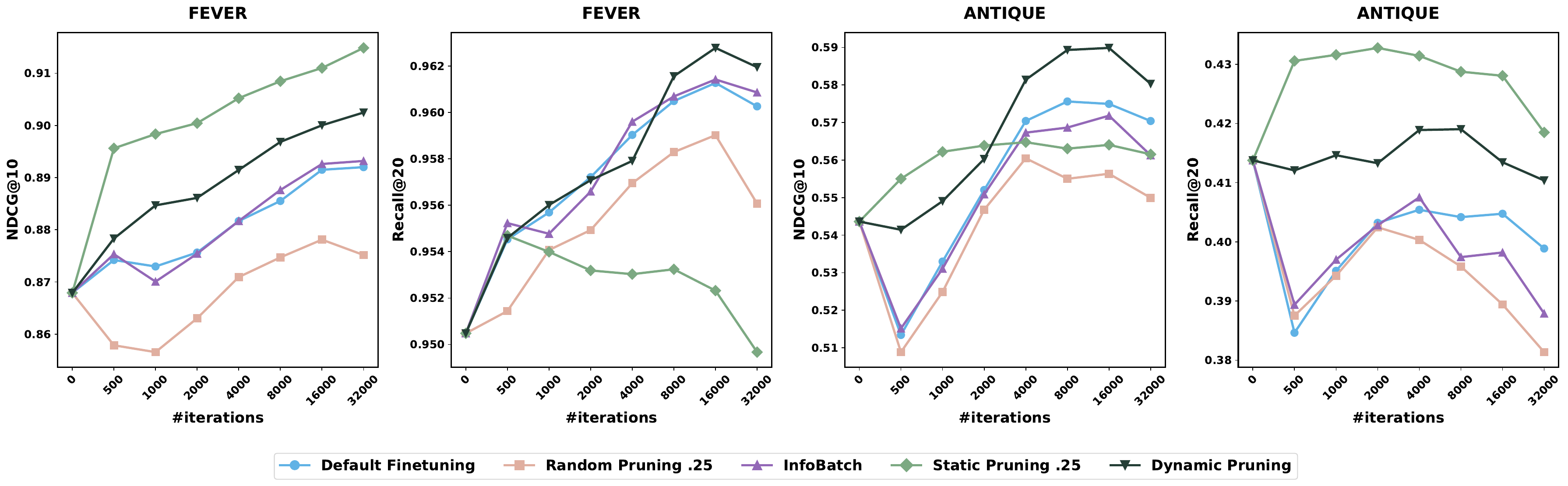}
    \caption{Training efficiency on ANTIQUE (unseen) and FEVER (seen). RP and SP use retention rate $k{=}0.25$.}
    \label{fig:efficiency}
\end{figure*}

Figure~\ref{fig:efficiency} illustrates the training efficiency on both the unseen ANTIQUE dataset~\cite{antique} and the FEVER dataset~\cite{fever} seen during pretraining. We evaluate SP and DP against FT and alternative pruning approaches across multiple training iterations. To ensure a fair comparison, we conducted separate experiments for each iteration count without using checkpoints, maintaining full learning rate scheduling in all experiments. \textbf{Note that for DP, each run produces a fundamentally different sampling trajectory due to the dependence of the cosine schedule on $t_{max}$, yet DP consistently outperforms baselines across all iteration counts, demonstrating robustness to these scheduling variations.}

Both pruning strategies demonstrate substantial efficiency gains. On FEVER, while FT and InfoBatch~\cite{infobatch} require 16,000 iterations to achieve optimal NDCG@10, DP achieves comparable results in fewer than 8,000 iterations, and SP reaches this level in fewer than 500 iterations. SP shows particularly rapid convergence in early training stages and consistently outperforms baseline approaches on NDCG@10, while DP demonstrates robustness across both NDCG@10 and Recall@20. Although DP introduces a small per-iteration overhead, this is offset by requiring fewer than 50\% of the training iterations to reach peak performance. On FEVER specifically, combining the $\le 8{,}000$ vs.\ $16{,}000$ iteration reduction with the $2.82\%$ per-step overhead at $I_u=10$ (Table~\ref{tab:computation_fever}) yields an approximately 49\% reduction in total wall-clock training time, making DP more efficient overall.

\subsubsection{Computation Overhead}
\label{section:overhead}

Table~\ref{tab:computation} compares the computational costs of DP with varying query update intervals ($I_u$). With the most frequent updates ($I_u = 1$), DP processes 2.32 iterations per second compared to FT's 2.43, a 4.49\% reduction. Setting $I_u = 100$ reduces this to 1.64\% while maintaining comparable performance, and $I_u = 10$ achieves the highest Recall@20 (0.646) with only 2.51\% overhead. The document update interval has negligible computational impact.

The overhead has two sources: (1) per-step sampling logic, which depends only on batch size and is negligible; and (2) score pooling and quantile computation at each update interval, which scales with the number of training pairs $N$ and inversely with $I_u$. Table~\ref{tab:computation} reports FiQA (14K pairs); on larger datasets the quantile computation is more expensive. We therefore also measure overhead on FEVER (109K queries, 140K pairs), one of our largest datasets (Table~\ref{tab:computation_fever}). While frequent updates ($I_u = 1$) incur 17.4\% overhead, a modest interval ($I_u = 10$) reduces it to only 2.82\% with on-par performance, confirming that DP remains efficient on large datasets. The current implementation uses \texttt{pad\_sequence}, \texttt{split}, and \texttt{quantile} for simplicity; the $I_u=1$ cost could be reduced substantially with further engineering, which we leave to future work given that $I_u=10$ already achieves negligible cost.

\begin{table}[t]
\centering
\caption{Per-step overhead of DP on FEVER, a large-scale dataset. The per-update score pooling and quantile computation scale with dataset size, but a modest update interval ($I_u{=}10$) reduces the overhead to $2.82\%$ with on-par performance, confirming scalability.}
\label{tab:computation_fever}
\resizebox{0.68\textwidth}{!}{
\begin{tabular}{lccc}
\toprule
\textbf{Update Interval} & \textbf{Per-step overhead} & \textbf{NDCG@10} & \textbf{Recall@20} \\
\midrule
FT & 0.00\% & 0.892 & 0.960 \\
DP ($I_u = 1$) & 17.40\% & \textbf{0.902} & 0.962 \\
DP ($I_u = 10$) & 2.82\% & 0.898 & \textbf{0.963} \\
\bottomrule
\end{tabular}
}
\end{table}

\begin{table*}[t]
\centering
\caption{Denoising evaluation on ANTIQUE with noisy positive samples. The FT baseline uses all training data ($k{=}1.0$); SP$_k$ then filters within this set at retention rate $k$, and SP$_k$+DP is the two-stage pipeline. Best in \textbf{bold}, second-best \underline{underlined}.}
\label{tab:denoising}
\resizebox{0.98\textwidth}{!}{
\begin{tabular}{l|ccc|cccc|ccc}
\toprule
\multirow{2}{*}{{\textbf{Metric}}} & \multicolumn{3}{c|}{{\textbf{Baselines}}} & \multicolumn{4}{c|}{{\textbf{OPERA (Individual)}}} & \multicolumn{3}{c}{{\textbf{OPERA (Two-Stage)}}} \\
\cmidrule(lr){2-4} \cmidrule(lr){5-8} \cmidrule(lr){9-11}
& \textbf{Pretrained} & \textbf{FT} & \textbf{InfoBatch} & \textbf{SP$_{.25}$} & \textbf{SP$_{.5}$} & \textbf{SP$_{.75}$} & \textbf{DP} & \textbf{SP$_{.25}$+DP} & \textbf{SP$_{.5}$+DP} & \textbf{SP$_{.75}$+DP} \\
\midrule
NDCG@10 & 0.543 & 0.570 & 0.567 & 0.560 & 0.560 & \underline{0.582} & \textbf{0.587} & 0.560 & 0.567 & \underline{0.582} \\
Recall@20 & 0.414 & 0.395 & 0.390 & \underline{0.430} & \underline{0.430} & 0.419 & 0.411 & 0.426 & \textbf{0.433} & 0.422 \\
\bottomrule
\end{tabular}
}
\end{table*}
\subsection{Denoising Capability}
\label{section:denoising}

We evaluate OPERA's robustness to label noise by introducing noisy samples into the ANTIQUE~\cite{antique} training set. Specifically, we include documents with lower relevance levels (level 2, which ``does not answer the question'') as positives, simulating real-world scenarios with imperfect annotations. The detailed setup is provided in Appendix~\ref{appendix:denoising_setup}.

Table~\ref{tab:denoising} presents the results. Consistent with Theorem~\ref{theorem:denoise}, which predicts that pruning outperforms FT when selection quality exceeds random, a clear divergence emerges between ranking and retrieval metrics under noisy conditions: while NDCG@10 improves from 0.543 (Pretrained) to 0.570 (FT), Recall@20 degrades from 0.414 to 0.395, indicating that retrieval effectiveness is more sensitive to false-positive training signals than ranking performance.

Among individual methods, DP achieves the highest NDCG@10 (0.587), while SP yields substantially better Recall@20 (0.430) than both FT (0.395) and DP (0.411). This advantage of SP in noisy settings motivates the two-stage approach: applying SP first to filter out noisy samples, followed by DP on the filtered data. Notably, SP$_{.75}$+DP outperforms all baselines on both metrics (NDCG@10: 0.582, Recall@20: 0.422), confirming that OPERA provides an effective denoising mechanism.

\subsection{Sampling Weight Visualization}
\label{subsec:sampling_vis}

\noindent
\begin{minipage}[t]{0.62\textwidth}
\vspace{0pt}
To analyze how DP allocates training focus over time, we visualize the evolution of sampling probabilities on the FiQA~\cite{fiqa} dataset (Figure~\ref{fig:pruning_vis}). Like self-paced learning~\cite{clsp}, DP continuously re-evaluates and adjusts sampling weights as model representations evolve, while ensuring all examples remain accessible throughout training (see Section~\ref{subsec:adaptive_training} for the detail comparison).

Figure~\ref{fig:pruning_vis}(a) shows the probability distribution grouped by query. Notably, nearly all queries maintain nonzero sampling probabilities throughout training, demonstrating that DP preserves broad query coverage unlike SP, which would exclude many queries entirely. The variation in intensity reflects quality-aware up-weighting while maintaining the diversity needed to preserve and even improve recall. Figure~\ref{fig:pruning_vis}(b), sorted by initial probabilities, reveals how DP dynamically redistributes attention: initially high-probability examples may decrease in importance while previously low-priority examples gain prominence. This demonstrates DP's ability to discover valuable training examples beyond what the pretrained model initially favors, shifting focus toward examples more informative in a domain-specific manner during the model training.
\end{minipage}%
\hfill
\begin{minipage}[t]{0.34\textwidth}
\vspace{-1pt}
\centering
\includegraphics[width=\linewidth]{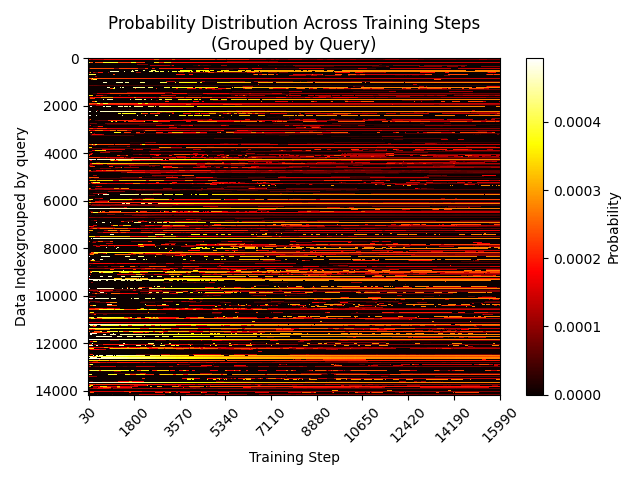}
{\small (a) Grouped by query}
\vspace{0pt}
\includegraphics[width=\linewidth]{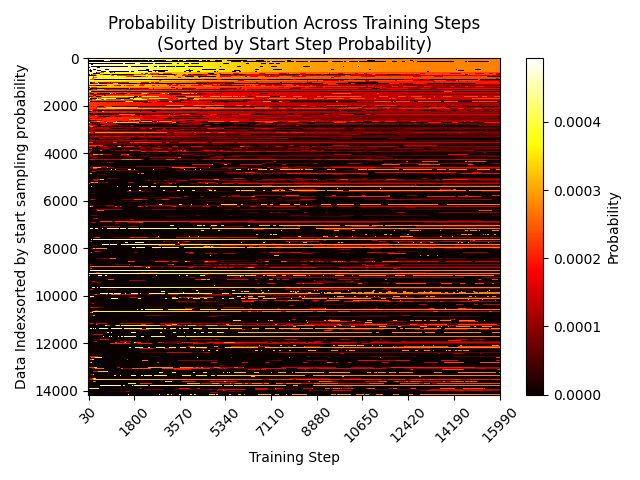}

{\small (b) Sorted by initial prob.}
\captionof{figure}{DP sampling probability evolution.}
\label{fig:pruning_vis}
\end{minipage}

\section{Conclusion}

We presented OPERA, a data pruning framework for domain adaptation of dense retrievers. Our investigation revealed a quality-coverage tradeoff intrinsic to the two-stage query-document sampling structure of retrieval training: static pruning (SP) improves ranking by focusing on high-quality pairs but reduces retrieval coverage, while dynamic pruning (DP) resolves this tradeoff through soft sampling modulation, achieving the best performance on both ranking and retrieval metrics across most evaluation settings while halving convergence time. Both approaches demonstrate effective denoising capability, especially when combined in a two-stage pipeline. Experiments on Qwen3-Embedding-0.6B provide evidence that these findings generalize beyond encoder-only architectures.

In practice, the choice depends on the application: SP is suited for ranking-focused scenarios where training speed is paramount, DP is preferred when both ranking and retrieval performance matter, and SP followed by DP is recommended when training data contains known label noise. We hope that the quality-coverage tradeoff identified in this work provides a useful lens for future research on efficient training for retrieval models.

\bibliography{custom}
\bibliographystyle{colm2026_conference}

\appendix
\newpage
\tableofcontents
\newpage

\section{Related Work}
\label{appendix:related_work}
\subsection{Dense Retrieval Models}
Recent advances in dense retrieval have demonstrated significant improvements over traditional sparse retrieval methods. These models leverage pretrained language models~\cite{bert, mistral, llama2} to generate robust embeddings in various retrieval tasks~\cite{mteb, beir}. While models like NV-Embed-v1~\cite{nvemb} with 7B parameters have pushed performance boundaries, a more compact model family such as BGE~\cite{bge, bgem3} offers an attractive balance between computational efficiency and effectiveness. Our work is based on the bge-large-en-v1.5 model~\cite{bge}, for its favorable balance between computational efficiency and strong performance, as well as its open source training and evaluation process, which enables reproducibility and practical deployment. This choice allows us to conduct extensive experiments on various data pruning strategies while maintaining manageable computational requirements.
More recently, LLM-based embedding models such as Qwen3-Embedding~\cite{qwen3embedding} have further advanced the state of the art by adapting large language models directly for embedding generation, achieving strong retrieval performance across diverse tasks, though at substantially higher computational cost.

\subsection{Data Pruning in Neural Networks}
Data pruning has emerged as a promising approach to improve training efficiency while maintaining or enhancing model performance~\cite{dp, dynamicp, el2n, pruningcvpr, llmdatacompos}. Recent innovations have shown that data pruning can surpass traditional power law scaling, improving efficiency with minimal performance degradation~\cite{beyond}. Previous work in this area can be categorized into two main approaches:

\paragraph{Static Pruning} Traditional static pruning methods select a fixed subset of training data before the training process begins~\cite{el2n, glister, coresets}. While these approaches can reduce training time, they often struggle with generalization across different architectures and datasets. In the context of dense retrievers, \cite{e5} introduced the consistency-based filter for pretraining, which retains only high-quality text pairs based on their ranking against random documents. Our work extends these ideas to the domain adaptation setting, where different considerations apply due to the distinct nature of the training data.

\paragraph{Dynamic Pruning} More recent approaches have explored dynamic data selection during training~\cite{dynamicp, dynamicu, infobatch, dpdiffusion, dpsample, dpspeech}. InfoBatch~\cite{infobatch} notably achieves training acceleration through unbiased dynamic pruning, adjusting sampling probabilities based on training loss. However, these methods primarily focus on maintaining model performance while improving convergence speed. In contrast, our approach demonstrates that careful selection of query-document pair can simultaneously enhance both training efficiency and model performance.

To the best of our knowledge, there is no prior work discussing data pruning for dense retrievers during the finetuning stage, where the query is selected regardless of its relevant documents. Therefore, we will compare our approach with our implementation of InfoBatch~\cite{infobatch} which prunes on the query level.

\paragraph{Diversity-based Selection} A related line of work selects training or retrieval examples to balance relevance and diversity, most notably Maximal Marginal Relevance (MMR)~\cite{mmr} and its successors, which iteratively pick items that are relevant yet dissimilar to those already chosen. While such methods explicitly optimize coverage, they are impractical for \emph{online}, per-step data selection during retriever finetuning. Selecting one document from $k$ candidates per query requires $O(\text{batch\_size} \times k^2)$ pairwise similarity computations to assess diversity, plus $O(\text{batch\_size} \times k)$ additional document embeddings per step, and this must be repeated as model scores drift throughout training. In contrast, OPERA reuses the cosine similarities already produced by the contrastive forward pass, adding only $1.6$--$4.5\%$ overhead (Table~\ref{tab:computation}) without any extra inference. OPERA instead achieves coverage structurally: by retaining every query and document with non-zero, dynamically scheduled sampling probability rather than by explicit diversity optimization. Reducing the cost of diversity-based selection for online retriever training is a promising direction we leave to future work.

\subsection{Curriculum Learning and Adaptive Training}
Curriculum learning has emerged as a promising approach to train neural networks by presenting training examples in a meaningful order. The idea of curriculum learning was formalized in~\cite{cl}, where it showed that gradually increasing the difficulty of training examples could lead to better generalization and faster convergence. In natural language processing, CL has been applied to machine translation~\cite{clmt}, question answering~\cite{clqa}, language modeling~\cite{cllm}, etc. Several approaches have been proposed for automatically determining the difficulty of training examples and generating curricula. Self-paced learning (SPL)~\cite{clsp} allows the model to automatically select its own curriculum based on the loss of training examples. Other works have explored the use of teacher-student frameworks~\cite{clts}, where a teacher model determines the curriculum for a student model. Related paradigms include active learning, which selects informative samples in discrete select-retrain cycles~\cite{al_survey}, and self-training, which iteratively refines pseudo-labels~\cite{self_training}. InfoBatch~\cite{infobatch} downsamples low-loss examples with inverse-probability reweighting.

Our dynamic pruning framework incorporates curriculum learning principles through its evolving threshold scheduling, which gradually adjusts both the ratio and sampling strength of training instances as the model's representations become more refined. However, DP differs from all of the above in three key respects: (1) it operates at two hierarchical levels (query and document), reflecting the two-stage sampling structure specific to retrieval training; (2) it uses increasing sampling strength via cosine scheduling, becoming more selective as the model's scoring improves, rather than relying on a fixed threshold or predetermined ordering; and (3) it maintains non-zero sampling probability for all examples throughout training, preserving coverage unlike curriculum learning and active learning which exclude examples entirely. A detailed comparison is provided in Section~\ref{subsec:adaptive_training}.

\section{Theoretical Analysis}
\label{appendix:proof}

We formalize the conditions under which data pruning outperforms standard finetuning. In dense retrieval, a query $q$ and its positive document $d$ are encoded into normalized embeddings $e_q, e_d \in \mathbb{R}^h$, with cosine similarity $s(e_q,e_d) = e_q^Te_d$. We analyze SP and DP by studying how each method's sampling strategy affects the optimal query embedding when some positive labels are noisy. Throughout the analysis, we consider a single query $q$.

\begin{lemma}\label{lem:monotone}
Let $u, v \in \mathbb{R}^n$ be unit vectors with $u \neq v$, and $k \in (0,1)$. Define
$f(k) = \frac{(ku + (1 - k)v)^T u}{\|ku + (1 - k)v\|}.$
Then $f'(k) > 0$ for all $k \in (0,1)$.

\noindent\textit{Proof.}
Direct computation yields
$f'(k) = \frac{(1 - (u^Tv)^2)(1 - k)}{\|ku + (1 - k)v\|^3}.$
Since $u \neq v \implies u^Tv < 1$, and $k < 1$, thus $f'(k) > 0$. \qed
\end{lemma}

\begin{theorem}\label{theorem:denoise}
Let $m_q$ be the total number of documents labeled as positive for query $q$, and $m_q^+$ be the number of correctly labeled documents ($m_q^+ \leq m_q$). Assume correctly labeled documents have embedding mean direction $\mu_1 \in \mathbb{R}^h$, noisy (false-positive) documents have mean direction $\mu_2 \in \mathbb{R}^h$, and negative documents have mean $0 \in \mathbb{R}^h$. For simplicity, we consider only the bias in the estimated query embedding and dismiss variance.

For SP, $r$ documents are selected ($\gamma r$ correctly labeled) for uniform sampling. For DP, $s$ documents are selected ($\rho s$ correctly labeled) with probability $\beta$ times higher than the remaining $m_q - s$ documents ($\beta > 1$). Define $E^{\text{FT}}$, $E^{\text{SP}}$, $E^{\text{DP}}$ as the expected cosine similarity between the optimal query embedding and true-positive documents under each strategy.

Then:
\begin{equation}
E^{SP} > E^{FT} \Leftrightarrow \gamma > \frac{m_q^+}{m_q}
\end{equation}
and similarly for DP. Additionally, when $\gamma = \rho$:
\begin{equation}
E^{SP} > E^{DP}.
\end{equation}
\end{theorem}

\begin{proof}
Under FT, the model maximizes $\mathcal{L}^{FT} = \sum_{i=1}^{m_q} s(e_q, e_{d_i})$, yielding:
\begin{align}
e_q^{\text{FT}} &= C_1\!\left(\tfrac{m_q^+}{m_q}\mu_1 + (1-\tfrac{m_q^+}{m_q})\mu_2\right)
\end{align}
and similarly for SP and DP:
\begin{align}
e_q^{\text{SP}} = C_2(\gamma\mu_1 + (1-\gamma)\mu_2), \quad
e_q^{\text{DP}} &= C_3((\rho s(\beta\!-\!1) + m_q^+)\mu_1 \nonumber \\
&\quad + ((1\!-\!\rho)s(\beta\!-\!1) + m_q\!-\!m_q^+)\mu_2)
\end{align}
where $C_1, C_2, C_3$ are normalization constants. By Lemma~\ref{lem:monotone}:
\begin{equation}
\gamma > \tfrac{m_q^+}{m_q}
\Leftrightarrow
s(e_q^{\rm SP},\mu_1)>s(e_q^{\rm FT},\mu_1)
\Leftrightarrow E^{SP} > E^{FT}
\end{equation}
and similarly for DP. Comparing SP and DP when $\gamma = \rho$:
\begin{equation}
E^{SP} > E^{DP} \Leftrightarrow \gamma > \tfrac{m_q^+}{m_q}
\end{equation}
\end{proof}

\section{Experimental Setup}
\label{appendix:setup}

\subsection{Dataset Statistics}
\label{appendix:dataset}

Table~\ref{tab:dataset-stats} presents detailed statistics for all eight evaluation datasets. The datasets vary significantly in scale, ranging from NFCorpus with 3,633 documents to TriviaQA with over 21 million documents. Training set sizes span from 2,426 queries (ANTIQUE) to 109,810 queries (FEVER), while the number of positive query-document pairs ranges from 14,166 (FiQA) to 741,436 (TriviaQA), reflecting diverse annotation densities across domains. FEVER and HotpotQA were previously seen by the pretrained model during its initial training, while the remaining six datasets are unseen.

\begin{table*}[h]
\centering
\caption{Dataset statistics. \#pos: number of positive query-document pairs. $^\dagger$Datasets seen during pretraining.}
\label{tab:dataset-stats}
\scalebox{0.85}{
\begin{tabular}{@{}lcr*{4}{r}@{}}
\toprule
\multirow{2}{*}{\textbf{Dataset}} & \multirow{2}{*}{\textbf{Domain}} & \multirow{2}{*}{\textbf{\#Docs}} & \multicolumn{2}{c}{\textbf{Train}} & \multicolumn{2}{c}{\textbf{Test}} \\
\cmidrule(lr){4-5} \cmidrule(lr){6-7}
& & & \textbf{\#q} & \textbf{\#pos} & \textbf{\#q} & \textbf{\#pos} \\
\midrule
NFCorpus & Nutrition & 3,633 & 2,590 & 110,575 & 323 & 12,334 \\
TripClick (h) & Medical & 1,523,878 & 3,529 & 55,663 & 1,175 & 32,067 \\
TripClick (t) & Medical & 1,523,878 & 105,964 & 424,820 & 1,175 & 6,202 \\
FiQA & Finance & 57,638 & 5,500 & 14,166 & 648 & 1,706 \\
ANTIQUE & Non-Factoid QA & 403,666 & 2,426 & 19,813 & 200 & 2,976 \\
TriviaQA & Factoid QA & 21,015,324 & 78,785 & 741,436 & 8,837 & 82,658 \\
HotpotQA$^\dagger$ & Factoid QA & 5,233,329 & 85,000 & 170,000 & 7,405 & 14,810 \\
FEVER$^\dagger$ & Fact Verif. & 5,416,568 & 109,810 & 140,085 & 6,666 & 7,937 \\
\bottomrule
\multicolumn{7}{l}{\footnotesize $^\dagger$ Datasets seen by the pretrained model.}
\end{tabular}
}
\end{table*}

\subsection{Dataset-specific Hyperparameters}
\label{appendix:hyperparams}

Table~\ref{tab:hyp_details} summarizes the dataset-specific hyperparameters used during training. All hyperparameters were optimized based on the vanilla FT baseline performance.

\begin{table}[h]
\centering
\caption{Dataset-specific hyperparameters, optimized for the FT baseline.}
\label{tab:hyp_details}
\scalebox{0.80}{
\begin{tabular}{lcccr}
\toprule
\textbf{Dataset} & \textbf{Negative} & \textbf{Mining} & $\mathbf{k}$ & \textbf{Training} \\
& \textbf{Mining} & \textbf{Range} & & \textbf{Iters} \\
\midrule
NFCorpus & Random & --- & 0.25 & 32,000 \\
TripClick (h) & Hard & 100--1,100 & 0.25 & 32,000 \\
TripClick (t) & Hard & 100--1,100 & 0.25 & 32,000 \\
FiQA & Hard & 10--100 & 0.50 & 8,000 \\
TriviaQA & Random & --- & 0.25 & 16,000 \\
ANTIQUE & Hard & 50--500 & 0.25 & 16,000 \\
FEVER & Hard & 10--200 & 0.25 & 32,000 \\
HotpotQA & Hard & 10--200 & 0.25 & 32,000 \\
\bottomrule
\end{tabular}
}
\end{table}

The model finetuning process uses a learning rate of 1e-6 with a linear scheduler and operates on a per-device train batch size of 8, yielding a total batch size of 64. Training uses FP16 precision with a maximum gradient norm of 1.0, no warmup, and no weight decay. The temperature is set to 0.02, and the system leverages cross-device negatives during training. The maximum query and passage length are set to 128 and 512 tokens, respectively. All embeddings are normalized. FiQA~\cite{fiqa}'s data retention ratio $k$ is set to 0.5 to accommodate its smaller dataset size.

For dynamic pruning, we initialize the query cutoff ratio ($r_s$) at 0.25, and use sampling strengths of 2 and 5 for the starting ($\alpha_s$) and ending ($\alpha_e$) values, respectively. The document level includes an initial cutoff ratio ($v_s$) of 0.25, a terminal cutoff ratio ($v_e$) of 0.5, and a constant sampling strength ($\beta$) of 5. All cutoff schedulers follow a cosine schedule.

\subsection{Qwen3-Embedding-0.6B Implementation Details}
\label{appendix:qwen3_impl}

For the Qwen3-Embedding-0.6B experiments, we use a learning rate of 1e-5 (10$\times$ higher than the BGE experiments) and train for 2,000 iterations across all datasets. The higher learning rate and reduced iteration count reflect the substantially higher computational cost of LLM-based retrievers and the potential data leakage from Qwen3-Embedding's large-scale pretraining corpus. We use a per-device batch size of 2, yielding a total batch size of 16, with gradient accumulation steps of 4 to achieve an effective batch size of 64. Training uses BF16 precision with the same temperature (0.02) and contrastive learning setup as the BGE experiments. The maximum query and passage lengths are set to 128 and 512 tokens, respectively, with instruction-based query encoding following the Qwen3-Embedding default prompts. The static pruning retention rate $k$ is set to 0.25 for all datasets. Dynamic pruning uses the same hyperparameters as the BGE experiments. Full results are provided in Appendix~\ref{appendix:qwen3_detailed}.

\subsection{Denoising Experiment Setup}
\label{appendix:denoising_setup}

The ANTIQUE dataset~\cite{antique} defines relevance levels as follows:
\begin{itemize}
    \item \textbf{Level 1:} Completely out of context or does not make any sense (4.6\% of the training data)
    \item \textbf{Level 2:} Does not answer the question, or provides an unreasonable answer, but is not out of context (23.1\% of the training data)
    \item \textbf{Level 3:} Can be an answer to the question, but is not sufficiently convincing (29.5\% of the training data)
    \item \textbf{Level 4:} Looks reasonable and convincing, with high quality (42.8\% of the training data)
\end{itemize}

Our experimental setup is designed as follows:
\begin{itemize}
    \item \textbf{Test set:} We maintain the standard evaluation criterion, considering only documents with relevance levels of 3 and 4 as positive samples~\cite{antique}.

    \item \textbf{Training set:} Documents with relevance levels of 2, 3, and 4 were treated as positive samples. The inclusion of level 2 documents, which are documented as insufficient answers, deliberately introduces noise into the positive samples. We did not include level 1 documents as positive samples since they represent a small portion (4.6\%) of the training data and are more like random noise, which differs from real-world scenarios where noisy samples are typically hard negatives that share some relevance with the query.
\end{itemize}


\section{Additional Experiments}
\label{appendix:additional}

\subsection{Consistency-Based Score Analysis}
\label{appendix:cbs}

\begin{table}[h]
\centering
\scalebox{1.}{
\begin{tabular}{lccc}
\toprule
\textbf{Dataset} & \textbf{Metric} & \textbf{SP} & \textbf{SP (CBS)} \\
\midrule
\multirow{2}{*}{NFCorpus} & NDCG@10 & \textbf{0.491} & 0.463 \\
& Recall@20 & 0.267 & \textbf{0.305} \\
\midrule
\multirow{2}{*}{FiQA} & NDCG@10 & 0.511 & \textbf{0.514} \\
& Recall@20 & 0.620 & \textbf{0.635} \\
\midrule
\multirow{2}{*}{ANTIQUE} & NDCG@10 & 0.564 & \textbf{0.577} \\
& Recall@20 & \textbf{0.428} & 0.403 \\
\midrule
\multirow{2}{*}{FEVER} & NDCG@10 & \textbf{0.915} & 0.894 \\
& Recall@20 & 0.950 & \textbf{0.961} \\
\midrule
\multirow{2}{*}{Average} & NDCG@10 & \textbf{0.620} & 0.612 \\
& Recall@20 & 0.566 & \textbf{0.576} \\
\bottomrule
\end{tabular}
}
\caption{SP with cosine similarity vs.\ CBS as pruning metric. Best per row in \textbf{bold}.}
\label{tab:consistency}
\end{table}

The consistency-based filter was introduced in~\cite{e5} as a quality control mechanism for large-scale pretraining. The core insight is that high-quality training pairs should maintain relevance compared to random documents: the method ranks each positive pair against random negatives and retains only those that consistently rank highly. We adapt this to our finetuning context with three modifications: (1) we leverage the pretrained checkpoint directly instead of training on noisy data, (2) we reduce the random document pool from one million to ten thousand to match the smaller scale of downstream tasks, and (3) we replace the static top-k threshold with a reciprocal rank metric (CBS) for percentage-based filtering. We sample random negatives from positive documents of other queries to avoid additional embedding computation.

We compare SP using CBS with SP using cosine similarity. Cosine similarity directly measures the semantic relationship between queries and documents, while CBS evaluates pairs based on their relative ranking against random samples.

Table~\ref{tab:consistency} presents the comparative performance across these datasets. SP uses cosine similarity as the similarity metric for pruning, while SP (CBS) uses CBS. In the nutrition dataset NFCorpus, SP achieves optimal NDCG@10 (0.491), while SP (CBS) leads in Recall@20 (0.305). FiQA demonstrates marginal improvements with SP (CBS) on both metrics (NDCG@10: 0.514, Recall@20: 0.635). For ANTIQUE, SP (CBS) leads in NDCG@10 (0.577), while SP with cosine similarity achieves higher Recall@20 (0.428). In FEVER, SP excels in NDCG@10 (0.915), while SP (CBS) shows superior Recall@20 (0.961). Overall, SP achieves better NDCG@10 (0.620) while SP (CBS) performs better on Recall@20 (0.576).

However, CBS is incompatible with DP as it requires computationally expensive recalculation of embeddings on all negative documents after each model update. We therefore adopt cosine similarity as the default pruning metric for both SP and DP.

\subsection{Static Pruning Data Retention Rate}
\label{appendix:sp_rate}

\begin{figure}[t]
    \centering
    \includegraphics[width=0.75\linewidth]{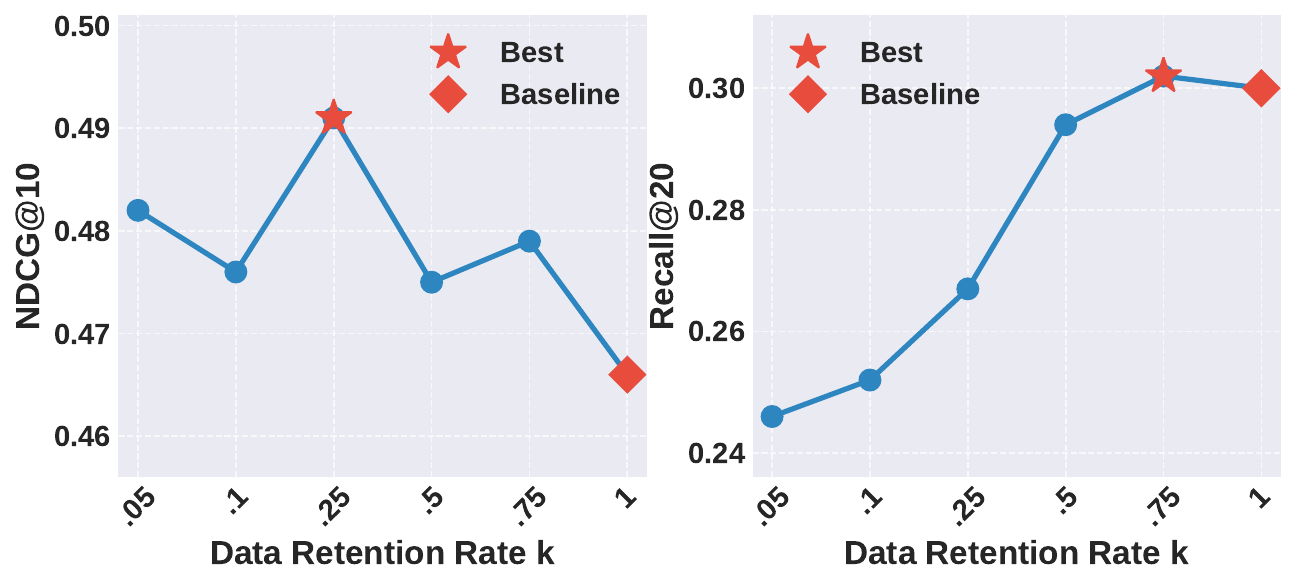}
    \caption{Effect of SP data retention rate $k$ on NFCorpus. SP improves NDCG@10 over FT at all retention rates.}
    \label{fig:pruning_ratio}
\end{figure}

We evaluate the effects of varying data retention rates ($k$) in Static Pruning (SP) through comprehensive experiments on NFCorpus~\cite{nfcorpus}. The experimental setup comprises six configurations: standard finetuning (FT) as baseline ($k=1$) and SP with $k=0.75,0.5,0.25,0.1,0.05$.

The results, presented in Figure~\ref{fig:pruning_ratio}, demonstrate that SP consistently surpasses FT in terms of NDCG@10 across all retention rates. The optimal performance was achieved at $k=0.25$, achieving a peak NDCG of 0.491. Notably, even with minimal data retention (5\%), SP maintains superior performance with an NDCG of 0.482 compared to FT (0.466). As expected, Recall@20 shows a monotonic decrease as the retention rate reduces from 100\% to 5\%, which aligns with the reduced data visibility. These findings highlight SP's data efficiency and indicate the potential for computational resource optimization while maintaining or enhancing ranking performance.

\subsection{Disentangling DP's Components}
\label{appendix:ablation_components}

\begin{table}[t]
\centering
\caption{Controlled single-component ablation of DP on FEVER. Each variant removes one component: \emph{Soft} (soft sampling vs.\ hard cutoff), \emph{Schedule} (cosine scheduling vs.\ fixed strength $\alpha{=}4,\beta{=}5,v{=}0.25$), and \emph{Hierarchy} (query/document two-level sampling). DP is the only variant that improves over FT on \emph{both} metrics. Best in \textbf{bold}.}
\label{tab:ablation_components}
\resizebox{0.72\textwidth}{!}{
\begin{tabular}{lccccc}
\toprule
\textbf{Variant} & \textbf{Hierarchy} & \textbf{Schedule} & \textbf{Soft} & \textbf{NDCG@10} & \textbf{Recall@20} \\
\midrule
FT & --- & --- & --- & 0.892 & 0.960 \\
SP & \xmark & \xmark & \xmark & \textbf{0.915} & 0.950 \\
DP (hard cutoff) & \cmark & \cmark & \xmark & 0.908 & 0.957 \\
DP (fixed-strength) & \cmark & \xmark & \cmark & 0.900 & 0.952 \\
\textbf{DP (full)} & \cmark & \cmark & \cmark & 0.902 & \textbf{0.962} \\
\bottomrule
\end{tabular}
}
\end{table}

The main-text ablation (Section~\ref{subsec:ablation_dp}) isolates DP's \emph{hierarchical} query/document design. DP additionally introduces two mechanisms: cosine \emph{scheduling} of the sampling strength, and \emph{soft} sampling that keeps low-scoring examples accessible with non-zero probability. To disentangle these, we run controlled single-component ablations on FEVER (Table~\ref{tab:ablation_components}), each removing one component while holding the others fixed:

\begin{itemize}
    \item \textbf{DP (hard cutoff)} removes soft sampling by hard-excluding below-threshold documents. NDCG improves (0.908) but Recall drops (0.957 vs.\ 0.962 for full DP): hard exclusion aids ranking at the cost of coverage.
    \item \textbf{DP (fixed-strength)} removes the schedule by fixing $\alpha{=}4,\beta{=}5,v{=}0.25$. Recall degrades to 0.952, below FT (0.960): without the self-correcting cosine schedule, performance becomes sensitive to the fixed strength values, whose suboptimality compounds over training.
    \item \textbf{SP} removes all three components (no hierarchy, no schedule, no soft sampling). It attains the best NDCG (0.915) but the worst Recall (0.950).
\end{itemize}

Full DP is the only variant that improves over FT on \emph{both} NDCG and Recall simultaneously, confirming that each component contributes to resolving the quality-coverage tradeoff rather than any single one being solely responsible for the gains.

\subsection{Additional Retriever Backbone: BGE-M3}
\label{appendix:bgem3}

\begin{table}[t]
\centering
\caption{OPERA on BGE-M3 (FiQA), a multilingual XLM-RoBERTa-based retriever with a different backbone, model size, and pretraining recipe from bge-large-en-v1.5. The same quality-coverage tradeoff reproduces: SP improves ranking but degrades recall, while DP improves both over FT. Best in \textbf{bold}.}
\label{tab:bgem3_results}
\resizebox{0.55\textwidth}{!}{
\begin{tabular}{llcc}
\toprule
\textbf{Model} & \textbf{Method} & \textbf{NDCG@10} & \textbf{Recall@20} \\
\midrule
\multirow{3}{*}{BGE-M3} & FT & 0.466 & 0.599 \\
& SP & \textbf{0.484} & 0.579 \\
& DP & 0.475 & \textbf{0.606} \\
\bottomrule
\end{tabular}
}
\end{table}

Top open-source retrievers on the MTEB leaderboard fall largely into two families that our two primary backbones already span: BERT-style encoders (bge-large-en-v1.5) and LLM-based embeddings (Qwen3-Embedding). To further probe generality, and following the reviewers' suggestion of a late-interaction style model, we add BGE-M3~\cite{bgem3} as a third backbone. BGE-M3 is built on XLM-RoBERTa-large, natively supports multi-vector (ColBERT-style) late interaction alongside dense and sparse retrieval, and differs from bge-large-en-v1.5 in training data (multilingual), model size (568M), and pretraining recipe. Table~\ref{tab:bgem3_results} reports FiQA results: the same quality-coverage tradeoff reproduces, with SP improving NDCG@10 (0.484 vs.\ 0.466) but degrading Recall@20 (0.579 vs.\ 0.599), while DP improves \emph{both} over FT (0.475 NDCG@10, 0.606 Recall@20). This confirms that OPERA's mechanism is not specific to a single backbone and generalizes to a multilingual model architecture. We use the same DP hyperparameters as the BGE experiments without additional tuning.

\subsection{Reversed-Schedule Ablation}
\label{appendix:reversed_schedule}

\begin{table}[h]
\centering
\caption{Effect of query sampling-strength schedule direction on NFCorpus. The default direction (broad$\to$selective, $\alpha$: $2\to5$) and the reversed direction ($\alpha$: $5\to2$) both outperform FT; the schedule direction acts as a precision/recall knob. Best in \textbf{bold}.}
\label{tab:reversed_schedule}
\resizebox{0.6\textwidth}{!}{
\begin{tabular}{lcc}
\toprule
\textbf{Method} & \textbf{NDCG@10} & \textbf{Recall@20} \\
\midrule
FT & 0.466 & 0.300 \\
SP & \textbf{0.491} & 0.267 \\
DP ($\alpha$: $2\to5$, default) & 0.480 & 0.304 \\
DP ($\alpha$: $5\to2$, reversed) & 0.476 & \textbf{0.310} \\
\bottomrule
\end{tabular}
}
\end{table}

DP's default schedule increases query sampling strength over training ($\alpha$: $2\to5$), reflecting growing trust in the model's quality signal (Section~\ref{subsec:dp}). To test sensitivity to the schedule direction, we compare against a reversed schedule ($\alpha$: $5\to2$, i.e., the standard curriculum direction of selective$\to$broad) on NFCorpus (Table~\ref{tab:reversed_schedule}). Both directions outperform FT on both metrics. The default (broad$\to$selective) direction achieves higher NDCG, supporting our intuition that sharpening selection toward high-quality examples as model scoring improves benefits ranking, while the reversed direction trades a little NDCG for slightly higher Recall by preserving more coverage late in training. This shows the method is robust to the schedule direction, and that the direction provides a meaningful precision/recall tradeoff knob.

\subsection{Per-Query Analysis by Difficulty}
\label{appendix:per_query}

\begin{table}[h]
\centering
\caption{Per-query analysis on FiQA (648 test queries) stratified into difficulty tertiles by the pretrained model's mean cosine similarity to positive documents (hard $\le 0.670 <$ medium $\le 0.717 <$ easy), evaluated at 4{,}000 iterations. DP improves over FT on hard and medium queries (where SP degrades), and both SP and DP improve on easy queries; DP's gains are distributed across all difficulty levels rather than concentrated on easy queries. Best in \textbf{bold}.}
\label{tab:per_query_difficulty}
\resizebox{0.9\textwidth}{!}{
\begin{tabular}{lc ccc @{\hspace{16pt}} ccc}
\toprule
& & \multicolumn{3}{c}{\textit{NDCG@10}} & \multicolumn{3}{c}{\textit{Recall@10}} \\
\cmidrule(lr){3-5} \cmidrule(lr){6-8}
\textbf{Bucket} & \textbf{n} & \textbf{FT} & \textbf{SP} & \textbf{DP} & \textbf{FT} & \textbf{SP} & \textbf{DP} \\
\midrule
Hard   & 216 & 0.209 & 0.203 & \textbf{0.212} & 0.254 & 0.250 & \textbf{0.254} \\
Medium & 216 & 0.438 & 0.435 & \textbf{0.447} & \textbf{0.524} & 0.520 & 0.522 \\
Easy   & 216 & 0.716 & \textbf{0.756} & 0.736 & 0.809 & 0.840 & \textbf{0.845} \\
\midrule
All    & 648 & 0.454 & \textbf{0.465} & \textbf{0.465} & 0.529 & 0.537 & \textbf{0.540} \\
\bottomrule
\end{tabular}
}
\end{table}

Average metrics could mask a scenario where pruning improves only ``easy'' queries at the expense of ``hard'' ones. To rule this out, we stratify the 648 FiQA test queries into difficulty tertiles by the pretrained model's mean cosine similarity to their positive documents (hard $\le 0.670 <$ medium $\le 0.717 <$ easy) and report per-bucket metrics at 4{,}000 iterations (Table~\ref{tab:per_query_difficulty}; not fully converged, so absolute differences are expected to grow with full training). DP improves over FT on hard and medium queries for NDCG@10 ($+0.003$, $+0.009$), where SP slightly degrades, confirming that DP does not sacrifice hard-query performance for easy-query gains. On easy queries, both SP and DP improve substantially over FT (SP strongest on NDCG, DP strongest on Recall). Overall, DP achieves the best NDCG@10 and Recall@10, with gains distributed across all difficulty levels rather than concentrated on a specific subset of queries.

\section{Detailed bge-large-en-v1.5 Results}
\label{appendix:bge_detailed}

We present full evaluation results for bge-large-en-v1.5 across all eight datasets at multiple training iterations (0 to 32,000). Each table reports NDCG, Recall, and Success at cutoffs of \{10, 20, 100\} for all methods: FT, RP (Random Pruning with the same retention rate as SP), InfoBatch, SP, and DP. The subscript in RP and SP denotes the data retention rate $k$ (e.g., RP25 retains 25\% of training pairs selected randomly, SP25 retains the top 25\% by similarity). FiQA additionally includes SP50 and RP50 variants due to its smaller dataset size. Results are shown for NFCorpus (Table~\ref{tab:nfcorpus_full}), FiQA (Table~\ref{tab:fiqa_full}), ANTIQUE (Table~\ref{tab:antique_full}), TriviaQA (Table~\ref{tab:triviaqa_full}), TripClick head (Table~\ref{tab:tripclick_head_full}), TripClick torso (Table~\ref{tab:tripclick_torso_full}), FEVER (Table~\ref{tab:fever_full}), and HotpotQA (Table~\ref{tab:hotpotqa_full}).

\section{Detailed Qwen3-Embedding-0.6B Results}
\label{appendix:qwen3_detailed}

We present comprehensive evaluation results for Qwen3-Embedding-0.6B across four datasets at 2,000 training iterations: NFCorpus (Table~\ref{tab:qwen3_detail_nfcorpus}), FiQA (Table~\ref{tab:qwen3_detail_fiqa}), ANTIQUE (Table~\ref{tab:qwen3_detail_antique}), and TriviaQA (Table~\ref{tab:qwen3_detail_triviaqa}). We report all evaluation metrics (MRR, NDCG, Recall, and Success at cutoffs of \{1, 5, 10, 20, 50, 100\}) for each configuration. The pretrained model results are iteration-independent and shown as a reference in each table. Notably, for FiQA (Table~\ref{tab:qwen3_detail_fiqa}), the pretrained model outperforms all finetuned methods across all metrics, suggesting that Qwen3-Embedding's pretraining corpus already contains high-quality financial domain data.

\section{Limitations}
\label{appendix:limitations}

We identify several directions for future work. Our LLM-based experiments use Qwen3-Embedding-0.6B as initial evidence; extending to larger models (e.g., 7B, 14B parameters) would strengthen the generalizability claim, though this is cost-prohibitive: both corpus indexing (embedding millions of documents) and training scale with model size, making an 8B-scale experiment suite roughly $10\times$ more expensive than our 0.6B setup. We note that the consistent quality-coverage tradeoff across three backbones spanning two architecture families (encoder-only BGE, decoder-based Qwen3, and multilingual XLM-RoBERTa-based BGE-M3; Appendix~\ref{appendix:bgem3}) suggests the underlying principle is architecture-agnostic, and the tradeoff arises from the two-stage sampling structure rather than model size. Our evaluation covers eight English-language text retrieval datasets across six domains; applying OPERA to genuinely multilingual retrieval benchmarks (e.g., MIRACL) and multimodal settings is a natural extension. We expect the quality-coverage tradeoff to be \emph{amplified} in multilingual settings, where high-quality pairs are likely concentrated in resource-rich languages while low-resource languages have sparser, noisier annotations; static pruning would disproportionately remove low-resource examples and degrade cross-lingual recall, whereas DP's soft modulation maintains exposure to all data while prioritizing high-quality pairs. While we report results under identical hyperparameters (optimized for the FT baseline of BGE) but use a higher learning rate for Qwen3 due to the substantially higher computational cost of LLM-based retrievers, we do not report variance across random seeds due to the computational cost of running all method-dataset combinations multiple times; however, the consistency of improvements across eight diverse datasets provides evidence of robustness beyond what single-dataset variance estimates would capture. Finally, OPERA's dynamic pruning hyperparameters (sampling strengths, cutoff ratios, and scheduling) were not tuned due to computational constraints; exploring different configurations or adaptive scheduling strategies may yield further gains.

\begin{table}
\centering
\resizebox{0.7\columnwidth}{!}{
\begin{tabular}{lccccccc}
\hline
Metric & \#iterations & FT & RP25 & InfoBatch & SP25 & DP \\
\hline
\multirow{8}{*}{NDCG@10} 
& 0 & 0.45101 & 0.45101 & 0.45101 & 0.45101 & 0.45101 \\
& 500 & 0.45608 & 0.45768 & 0.46132 & 0.47209 & 0.46939 \\
& 1000 & 0.46758 & 0.46583 & 0.46669 & 0.47429 & 0.47479 \\
& 2000 & 0.47059 & 0.46631 & 0.47209 & 0.47979 & 0.47418 \\
& 4000 & 0.47465 & 0.47516 & 0.47438 & 0.48804 & 0.48388 \\
& 8000 & 0.47471 & 0.45486 & 0.47866 & 0.49103 & 0.49008 \\
& 16000 & 0.4802 & 0.43449 & 0.48371 & 0.48671 & 0.48067 \\
& 32000 & 0.46565 & 0.42648 & 0.46965 & 0.49128 & 0.48019 \\
\hline
\multirow{8}{*}{Recall@10} 
& 0 & 0.18786 & 0.18786 & 0.18786 & 0.18786 & 0.18786 \\
& 500 & 0.19205 & 0.19057 & 0.1934 & 0.19968 & 0.19693 \\
& 1000 & 0.20346 & 0.19951 & 0.20231 & 0.198 & 0.20538 \\
& 2000 & 0.20952 & 0.20231 & 0.20912 & 0.20422 & 0.20623 \\
& 4000 & 0.21562 & 0.20594 & 0.21102 & 0.20835 & 0.21765 \\
& 8000 & 0.21083 & 0.20026 & 0.21204 & 0.21186 & 0.21774 \\
& 16000 & 0.21814 & 0.18932 & 0.22421 & 0.20753 & 0.22189 \\
& 32000 & 0.22654 & 0.18213 & 0.23186 & 0.21049 & 0.2272 \\
\hline
\multirow{8}{*}{Success@10} 
& 0 & 0.75542 & 0.75542 & 0.75542 & 0.75542 & 0.75542 \\
& 500 & 0.74923 & 0.73994 & 0.74613 & 0.75851 & 0.75851 \\
& 1000 & 0.75542 & 0.74613 & 0.75232 & 0.75542 & 0.76471 \\
& 2000 & 0.76161 & 0.75542 & 0.76161 & 0.76471 & 0.76161 \\
& 4000 & 0.7709 & 0.77709 & 0.76471 & 0.77709 & 0.77709 \\
& 8000 & 0.77709 & 0.76471 & 0.78019 & 0.78019 & 0.7678 \\
& 16000 & 0.76471 & 0.73684 & 0.76161 & 0.76471 & 0.78019 \\
& 32000 & 0.75851 & 0.72446 & 0.74923 & 0.75232 & 0.74613 \\
\hline
\multirow{8}{*}{NDCG@20} 
& 0 & 0.46963 & 0.46963 & 0.46963 & 0.46963 & 0.46963 \\
& 500 & 0.47225 & 0.47045 & 0.47562 & 0.48777 & 0.4863 \\
& 1000 & 0.48135 & 0.47612 & 0.48132 & 0.493 & 0.49069 \\
& 2000 & 0.48129 & 0.48145 & 0.4836 & 0.50005 & 0.48827 \\
& 4000 & 0.48389 & 0.48097 & 0.48472 & 0.50418 & 0.49289 \\
& 8000 & 0.48058 & 0.46252 & 0.48539 & 0.50196 & 0.50145 \\
& 16000 & 0.48987 & 0.44691 & 0.49305 & 0.50082 & 0.49279 \\
& 32000 & 0.476 & 0.43441 & 0.48058 & 0.4989 & 0.49261 \\
\hline
\multirow{8}{*}{Recall@20} 
& 0 & 0.23235 & 0.23235 & 0.23235 & 0.23235 & 0.23235 \\
& 500 & 0.24643 & 0.23986 & 0.24754 & 0.24674 & 0.25311 \\
& 1000 & 0.25214 & 0.24472 & 0.2516 & 0.25396 & 0.25597 \\
& 2000 & 0.2572 & 0.25875 & 0.25863 & 0.25857 & 0.25766 \\
& 4000 & 0.2667 & 0.25901 & 0.26489 & 0.26086 & 0.26595 \\
& 8000 & 0.27199 & 0.25249 & 0.27128 & 0.25804 & 0.27595 \\
& 16000 & 0.29227 & 0.24643 & 0.29663 & 0.26297 & 0.292 \\
& 32000 & 0.30031 & 0.24182 & 0.30361 & 0.2671 & 0.30398 \\
\hline
\multirow{8}{*}{Success@20} 
& 0 & 0.80495 & 0.80495 & 0.80495 & 0.80495 & 0.80495 \\
& 500 & 0.79876 & 0.79567 & 0.80186 & 0.80186 & 0.81115 \\
& 1000 & 0.79567 & 0.78947 & 0.79567 & 0.81734 & 0.80805 \\
& 2000 & 0.79876 & 0.80495 & 0.79876 & 0.81734 & 0.80186 \\
& 4000 & 0.81734 & 0.80805 & 0.80805 & 0.80805 & 0.81424 \\
& 8000 & 0.81115 & 0.80805 & 0.82353 & 0.81424 & 0.81424 \\
& 16000 & 0.82353 & 0.78328 & 0.82663 & 0.81115 & 0.81734 \\
& 32000 & 0.79567 & 0.76161 & 0.79257 & 0.80186 & 0.80805 \\
\hline
\multirow{8}{*}{NDCG@100} 
& 0 & 0.5908 & 0.5908 & 0.5908 & 0.5908 & 0.5908 \\
& 500 & 0.59424 & 0.59384 & 0.59712 & 0.60724 & 0.60158 \\
& 1000 & 0.60187 & 0.5992 & 0.60232 & 0.60983 & 0.60755 \\
& 2000 & 0.60681 & 0.60372 & 0.60681 & 0.61762 & 0.6134 \\
& 4000 & 0.61242 & 0.61026 & 0.61227 & 0.62192 & 0.62047 \\
& 8000 & 0.61675 & 0.5967 & 0.62023 & 0.6202 & 0.62778 \\
& 16000 & 0.61673 & 0.58437 & 0.61839 & 0.62012 & 0.61766 \\
& 32000 & 0.60532 & 0.57457 & 0.60861 & 0.61819 & 0.61702 \\
\hline
\multirow{8}{*}{Recall@100} 
& 0 & 0.36386 & 0.36386 & 0.36386 & 0.36386 & 0.36386 \\
& 500 & 0.39106 & 0.39251 & 0.39036 & 0.38694 & 0.39014 \\
& 1000 & 0.40479 & 0.4013 & 0.40445 & 0.3935 & 0.40153 \\
& 2000 & 0.4134 & 0.40905 & 0.41078 & 0.39942 & 0.41369 \\
& 4000 & 0.43184 & 0.42072 & 0.43176 & 0.40591 & 0.4305 \\
& 8000 & 0.46136 & 0.43142 & 0.46095 & 0.40966 & 0.45188 \\
& 16000 & 0.48035 & 0.43325 & 0.47794 & 0.41957 & 0.47781 \\
& 32000 & 0.50295 & 0.42943 & 0.49731 & 0.42292 & 0.49686 \\
\hline
\multirow{8}{*}{Success@100} 
& 0 & 0.89474 & 0.89474 & 0.89474 & 0.89474 & 0.89474 \\
& 500 & 0.89783 & 0.89474 & 0.89783 & 0.89474 & 0.88854 \\
& 1000 & 0.89164 & 0.89474 & 0.89474 & 0.89783 & 0.88854 \\
& 2000 & 0.90093 & 0.89783 & 0.89474 & 0.90402 & 0.89783 \\
& 4000 & 0.90402 & 0.90093 & 0.89474 & 0.89783 & 0.90712 \\
& 8000 & 0.89474 & 0.89474 & 0.89474 & 0.88545 & 0.89783 \\
& 16000 & 0.88235 & 0.89164 & 0.87926 & 0.89783 & 0.89783 \\
& 32000 & 0.88854 & 0.88854 & 0.88545 & 0.89164 & 0.89783 \\
\hline
\end{tabular}
}
\caption{Full Results on NFCorpus.}
\label{tab:nfcorpus_full}
\end{table}
\begin{table}
\centering
\resizebox{0.7\columnwidth}{!}{
\begin{tabular}{lccccccc}
\hline
Metric & \#iterations & FT & RP50 & InfoBatch & SP25 & SP50 & DP \\
\hline
\multirow{7}{*}{NDCG@10} 
& 0 & 0.48919 & 0.48919 & 0.48919 & 0.48919 & 0.48919 & 0.48919 \\
& 500 & 0.5073 & 0.50502 & 0.50971 & 0.50296 & 0.50469 & 0.50404 \\
& 1000 & 0.5085 & 0.50659 & 0.50761 & 0.50241 & 0.50876 & 0.50655 \\
& 2000 & 0.51277 & 0.50653 & 0.51195 & 0.50229 & 0.51293 & 0.51212 \\
& 4000 & 0.5148 & 0.5026 & 0.51559 & 0.51047 & 0.51555 & 0.51812 \\
& 8000 & 0.51359 & 0.50572 & 0.5137 & 0.511 & 0.5162 & 0.52418 \\
& 16000 & 0.50859 & 0.49985 & 0.50758 & 0.51213 & 0.51463 & 0.51715 \\
\hline
\multirow{7}{*}{Recall@10} 
& 0 & 0.51432 & 0.51432 & 0.51432 & 0.51432 & 0.51432 & 0.51432 \\
& 500 & 0.53982 & 0.53969 & 0.53958 & 0.53143 & 0.5367 & 0.54259 \\
& 1000 & 0.54812 & 0.5451 & 0.54982 & 0.53329 & 0.54211 & 0.54512 \\
& 2000 & 0.559 & 0.54228 & 0.55506 & 0.53052 & 0.54561 & 0.5546 \\
& 4000 & 0.55662 & 0.54005 & 0.55558 & 0.54067 & 0.54815 & 0.56327 \\
& 8000 & 0.54711 & 0.53811 & 0.54964 & 0.53731 & 0.54906 & 0.55887 \\
& 16000 & 0.54124 & 0.52443 & 0.53852 & 0.54046 & 0.5469 & 0.55478 \\
\hline
\multirow{7}{*}{Success@10} 
& 0 & 0.72685 & 0.72685 & 0.72685 & 0.72685 & 0.72685 & 0.72685 \\
& 500 & 0.74383 & 0.74383 & 0.74691 & 0.73611 & 0.74228 & 0.74846 \\
& 1000 & 0.75154 & 0.74537 & 0.75 & 0.74074 & 0.74691 & 0.74846 \\
& 2000 & 0.7608 & 0.74383 & 0.75772 & 0.73457 & 0.74846 & 0.7608 \\
& 4000 & 0.76235 & 0.75 & 0.76389 & 0.74074 & 0.75 & 0.76543 \\
& 8000 & 0.75463 & 0.75154 & 0.75926 & 0.73457 & 0.74537 & 0.76389 \\
& 16000 & 0.74691 & 0.74074 & 0.74537 & 0.73765 & 0.74383 & 0.76389 \\
\hline
\multirow{7}{*}{NDCG@20} 
& 0 & 0.52218 & 0.52218 & 0.52218 & 0.52218 & 0.52218 & 0.52218 \\
& 500 & 0.54164 & 0.5384 & 0.54391 & 0.53541 & 0.53779 & 0.53734 \\
& 1000 & 0.54275 & 0.54106 & 0.54071 & 0.53541 & 0.54142 & 0.54203 \\
& 2000 & 0.54417 & 0.54016 & 0.54539 & 0.53641 & 0.54587 & 0.54416 \\
& 4000 & 0.54578 & 0.53642 & 0.54629 & 0.54136 & 0.54946 & 0.54855 \\
& 8000 & 0.54723 & 0.54071 & 0.54671 & 0.54296 & 0.54924 & 0.5539 \\
& 16000 & 0.54209 & 0.53193 & 0.54116 & 0.54289 & 0.54821 & 0.54795 \\
\hline
\multirow{7}{*}{Recall@20} 
& 0 & 0.60161 & 0.60161 & 0.60161 & 0.60161 & 0.60161 & 0.60161 \\
& 500 & 0.63196 & 0.62872 & 0.63162 & 0.61578 & 0.62241 & 0.63024 \\
& 1000 & 0.6369 & 0.63292 & 0.6364 & 0.61609 & 0.62523 & 0.63847 \\
& 2000 & 0.63903 & 0.62874 & 0.64027 & 0.61811 & 0.62882 & 0.63727 \\
& 4000 & 0.63808 & 0.63194 & 0.63514 & 0.62047 & 0.63412 & 0.64222 \\
& 8000 & 0.63933 & 0.63279 & 0.63994 & 0.61994 & 0.63041 & 0.63909 \\
& 16000 & 0.62971 & 0.61118 & 0.63012 & 0.61941 & 0.63243 & 0.63809 \\
\hline
\multirow{7}{*}{Success@20} 
& 0 & 0.79321 & 0.79321 & 0.79321 & 0.79321 & 0.79321 & 0.79321 \\
& 500 & 0.81636 & 0.81327 & 0.81636 & 0.80093 & 0.81019 & 0.81481 \\
& 1000 & 0.82407 & 0.82253 & 0.82407 & 0.80401 & 0.81019 & 0.82716 \\
& 2000 & 0.82562 & 0.81327 & 0.82407 & 0.80556 & 0.81944 & 0.82562 \\
& 4000 & 0.82099 & 0.81481 & 0.8179 & 0.8071 & 0.82407 & 0.82407 \\
& 8000 & 0.8287 & 0.81327 & 0.82562 & 0.80864 & 0.81944 & 0.82099 \\
& 16000 & 0.8179 & 0.80247 & 0.81481 & 0.8071 & 0.8179 & 0.82253 \\
\hline
\multirow{7}{*}{NDCG@100} 
& 0 & 0.57052 & 0.57052 & 0.57052 & 0.57052 & 0.57052 & 0.57052 \\
& 500 & 0.58557 & 0.58211 & 0.58748 & 0.57741 & 0.58011 & 0.58142 \\
& 1000 & 0.58628 & 0.58487 & 0.58404 & 0.57796 & 0.58414 & 0.58462 \\
& 2000 & 0.58806 & 0.58638 & 0.58894 & 0.57922 & 0.58792 & 0.58797 \\
& 4000 & 0.59107 & 0.58097 & 0.59279 & 0.58342 & 0.59066 & 0.59249 \\
& 8000 & 0.59063 & 0.58115 & 0.5907 & 0.58538 & 0.59152 & 0.59951 \\
& 16000 & 0.58399 & 0.57598 & 0.58324 & 0.58426 & 0.58982 & 0.59 \\
\hline
\multirow{7}{*}{Recall@100} 
& 0 & 0.77048 & 0.77048 & 0.77048 & 0.77048 & 0.77048 & 0.77048 \\
& 500 & 0.78785 & 0.78461 & 0.78634 & 0.75805 & 0.76858 & 0.78739 \\
& 1000 & 0.79259 & 0.79034 & 0.79133 & 0.76416 & 0.77228 & 0.7897 \\
& 2000 & 0.7951 & 0.7905 & 0.79659 & 0.76461 & 0.77901 & 0.79448 \\
& 4000 & 0.79596 & 0.78875 & 0.79938 & 0.76684 & 0.78161 & 0.79421 \\
& 8000 & 0.79368 & 0.77385 & 0.79628 & 0.76688 & 0.78176 & 0.80139 \\
& 16000 & 0.7764 & 0.76227 & 0.77643 & 0.76078 & 0.77788 & 0.78376 \\
\hline
\multirow{7}{*}{Success@100} 
& 0 & 0.89506 & 0.89506 & 0.89506 & 0.89506 & 0.89506 & 0.89506 \\
& 500 & 0.90586 & 0.89969 & 0.90432 & 0.89352 & 0.89815 & 0.90586 \\
& 1000 & 0.90895 & 0.90586 & 0.90586 & 0.89506 & 0.89969 & 0.90432 \\
& 2000 & 0.90895 & 0.90741 & 0.90895 & 0.89352 & 0.90432 & 0.90586 \\
& 4000 & 0.90895 & 0.90278 & 0.91204 & 0.89815 & 0.90432 & 0.90741 \\
& 8000 & 0.91049 & 0.8966 & 0.91358 & 0.89969 & 0.90432 & 0.91358 \\
& 16000 & 0.90123 & 0.8966 & 0.89969 & 0.89506 & 0.90278 & 0.90586 \\
\hline
\end{tabular}
}
\caption{Full Results on FiQA.}
\label{tab:fiqa_full}
\end{table}
\begin{table}
\centering
\resizebox{0.7\columnwidth}{!}{
\begin{tabular}{lccccccc}
\hline
Metric & \#iterations & FT & RP25 & InfoBatch & SP25 & DP \\
\hline
\multirow{8}{*}{NDCG@10} 
& 0 & 0.54361 & 0.54361 & 0.54361 & 0.54361 & 0.54361 \\
& 500 & 0.51341 & 0.50884 & 0.51514 & 0.555 & 0.54143 \\
& 1000 & 0.53297 & 0.52488 & 0.53106 & 0.56221 & 0.54905 \\
& 2000 & 0.55202 & 0.54677 & 0.55086 & 0.56385 & 0.56028 \\
& 4000 & 0.57041 & 0.56046 & 0.5673 & 0.56475 & 0.5814 \\
& 8000 & 0.5756 & 0.55501 & 0.56865 & 0.56304 & 0.58932 \\
& 16000 & 0.57495 & 0.55635 & 0.57182 & 0.56404 & 0.58987 \\
& 32000 & 0.5704 & 0.54994 & 0.56128 & 0.56153 & 0.5803 \\
\hline
\multirow{8}{*}{Recall@10} 
& 0 & 0.3187 & 0.3187 & 0.3187 & 0.3187 & 0.3187 \\
& 500 & 0.30215 & 0.30074 & 0.30394 & 0.32731 & 0.3211 \\
& 1000 & 0.31274 & 0.30605 & 0.31297 & 0.33135 & 0.32435 \\
& 2000 & 0.31647 & 0.3128 & 0.31478 & 0.33139 & 0.32999 \\
& 4000 & 0.32797 & 0.31715 & 0.32603 & 0.33024 & 0.33619 \\
& 8000 & 0.32685 & 0.30934 & 0.32143 & 0.32961 & 0.33503 \\
& 16000 & 0.32654 & 0.30758 & 0.32119 & 0.33087 & 0.33564 \\
& 32000 & 0.3201 & 0.303 & 0.31485 & 0.32149 & 0.3269 \\
\hline
\multirow{8}{*}{Success@10} 
& 0 & 0.915 & 0.915 & 0.915 & 0.915 & 0.915 \\
& 500 & 0.915 & 0.915 & 0.91 & 0.925 & 0.92 \\
& 1000 & 0.92 & 0.93 & 0.925 & 0.925 & 0.92 \\
& 2000 & 0.93 & 0.93 & 0.93 & 0.925 & 0.925 \\
& 4000 & 0.93 & 0.93 & 0.93 & 0.925 & 0.94 \\
& 8000 & 0.94 & 0.925 & 0.935 & 0.925 & 0.935 \\
& 16000 & 0.94 & 0.92 & 0.93 & 0.925 & 0.94 \\
& 32000 & 0.93 & 0.92 & 0.935 & 0.92 & 0.935 \\
\hline
\multirow{8}{*}{NDCG@20} 
& 0 & 0.58938 & 0.58938 & 0.58938 & 0.58938 & 0.58938 \\
& 500 & 0.55908 & 0.55891 & 0.56392 & 0.60371 & 0.58903 \\
& 1000 & 0.58022 & 0.57683 & 0.57841 & 0.60726 & 0.59685 \\
& 2000 & 0.60082 & 0.59966 & 0.59985 & 0.60895 & 0.60292 \\
& 4000 & 0.61127 & 0.60587 & 0.61113 & 0.61126 & 0.62214 \\
& 8000 & 0.61673 & 0.6027 & 0.61021 & 0.60916 & 0.63131 \\
& 16000 & 0.61622 & 0.60262 & 0.61349 & 0.60949 & 0.62942 \\
& 32000 & 0.61282 & 0.59419 & 0.60292 & 0.60826 & 0.6267 \\
\hline
\multirow{8}{*}{Recall@20} 
& 0 & 0.41377 & 0.41377 & 0.41377 & 0.41377 & 0.41377 \\
& 500 & 0.38463 & 0.38752 & 0.38937 & 0.43058 & 0.41207 \\
& 1000 & 0.3951 & 0.39427 & 0.39699 & 0.43161 & 0.41464 \\
& 2000 & 0.40321 & 0.40246 & 0.4028 & 0.43278 & 0.41329 \\
& 4000 & 0.40542 & 0.40033 & 0.40752 & 0.43143 & 0.4189 \\
& 8000 & 0.40418 & 0.39583 & 0.39742 & 0.42876 & 0.41904 \\
& 16000 & 0.40475 & 0.38942 & 0.3982 & 0.42807 & 0.41345 \\
& 32000 & 0.39887 & 0.38136 & 0.38788 & 0.41849 & 0.41037 \\
\hline
\multirow{8}{*}{Success@20} 
& 0 & 0.94 & 0.94 & 0.94 & 0.94 & 0.94 \\
& 500 & 0.935 & 0.935 & 0.93 & 0.945 & 0.93 \\
& 1000 & 0.93 & 0.93 & 0.935 & 0.95 & 0.935 \\
& 2000 & 0.94 & 0.935 & 0.945 & 0.95 & 0.945 \\
& 4000 & 0.95 & 0.95 & 0.95 & 0.955 & 0.95 \\
& 8000 & 0.95 & 0.95 & 0.955 & 0.95 & 0.955 \\
& 16000 & 0.945 & 0.94 & 0.955 & 0.95 & 0.955 \\
& 32000 & 0.945 & 0.935 & 0.95 & 0.94 & 0.955 \\
\hline
\multirow{8}{*}{NDCG@100} 
& 0 & 0.70963 & 0.70963 & 0.70963 & 0.70963 & 0.70963 \\
& 500 & 0.68968 & 0.68743 & 0.69259 & 0.72254 & 0.71339 \\
& 1000 & 0.70339 & 0.69928 & 0.70178 & 0.72514 & 0.72015 \\
& 2000 & 0.71801 & 0.71431 & 0.71782 & 0.72624 & 0.72481 \\
& 4000 & 0.72655 & 0.71896 & 0.72348 & 0.72845 & 0.73586 \\
& 8000 & 0.72726 & 0.71366 & 0.7236 & 0.72757 & 0.73949 \\
& 16000 & 0.72375 & 0.71325 & 0.72022 & 0.72788 & 0.73581 \\
& 32000 & 0.71905 & 0.70656 & 0.70943 & 0.72737 & 0.73014 \\
\hline
\multirow{8}{*}{Recall@100} 
& 0 & 0.61225 & 0.61225 & 0.61225 & 0.61225 & 0.61225 \\
& 500 & 0.59569 & 0.59606 & 0.60046 & 0.62704 & 0.61638 \\
& 1000 & 0.59057 & 0.59225 & 0.59289 & 0.62929 & 0.61486 \\
& 2000 & 0.59087 & 0.58942 & 0.58999 & 0.63046 & 0.61596 \\
& 4000 & 0.58692 & 0.58056 & 0.58274 & 0.63033 & 0.60371 \\
& 8000 & 0.57344 & 0.56642 & 0.5742 & 0.63086 & 0.59186 \\
& 16000 & 0.56939 & 0.55861 & 0.5599 & 0.62933 & 0.58018 \\
& 32000 & 0.56019 & 0.55254 & 0.55349 & 0.61963 & 0.57318 \\
\hline
\multirow{8}{*}{Success@100} 
& 0 & 0.97 & 0.97 & 0.97 & 0.97 & 0.97 \\
& 500 & 0.97 & 0.975 & 0.965 & 0.975 & 0.98 \\
& 1000 & 0.975 & 0.975 & 0.975 & 0.975 & 0.98 \\
& 2000 & 0.975 & 0.97 & 0.975 & 0.975 & 0.975 \\
& 4000 & 0.975 & 0.97 & 0.975 & 0.975 & 0.975 \\
& 8000 & 0.975 & 0.97 & 0.975 & 0.975 & 0.975 \\
& 16000 & 0.975 & 0.965 & 0.97 & 0.975 & 0.975 \\
& 32000 & 0.97 & 0.965 & 0.96 & 0.97 & 0.97 \\
\hline
\end{tabular}
}
\caption{Full Results on ANTIQUE.}
\label{tab:antique_full}
\end{table}
\begin{table}
\centering
\resizebox{0.7\columnwidth}{!}{
\begin{tabular}{lccccccc}
\hline
Metric & \#iterations & FT & RP25 & InfoBatch & SP25 & DP \\
\hline
\multirow{8}{*}{NDCG@10} 
& 0 & 0.48724 & 0.48724 & 0.48724 & 0.48724 & 0.48724 \\
& 500 & 0.48569 & 0.48454 & 0.48815 & 0.49358 & 0.49215 \\
& 1000 & 0.48447 & 0.48167 & 0.48317 & 0.49497 & 0.49271 \\
& 2000 & 0.48155 & 0.47685 & 0.48055 & 0.4958 & 0.49108 \\
& 4000 & 0.47759 & 0.4755 & 0.47892 & 0.49768 & 0.49108 \\
& 8000 & 0.47752 & 0.47493 & 0.47722 & 0.49818 & 0.48562 \\
& 16000 & 0.48116 & 0.47636 & 0.48179 & 0.50124 & 0.49078 \\
& 32000 & 0.48315 & 0.47358 & 0.48287 & 0.50933 & 0.49009 \\
\hline
\multirow{8}{*}{Recall@10} 
& 0 & 0.3418 & 0.3418 & 0.3418 & 0.3418 & 0.3418 \\
& 500 & 0.35714 & 0.3568 & 0.35926 & 0.35052 & 0.35995 \\
& 1000 & 0.35744 & 0.35596 & 0.35722 & 0.35049 & 0.3616 \\
& 2000 & 0.35736 & 0.35347 & 0.35697 & 0.3521 & 0.36049 \\
& 4000 & 0.35613 & 0.35363 & 0.35553 & 0.3541 & 0.36056 \\
& 8000 & 0.3566 & 0.35452 & 0.35697 & 0.35376 & 0.35896 \\
& 16000 & 0.3609 & 0.35589 & 0.36158 & 0.35759 & 0.36408 \\
& 32000 & 0.36281 & 0.35282 & 0.36419 & 0.36243 & 0.36303 \\
\hline
\multirow{8}{*}{Success@10} 
& 0 & 0.81775 & 0.81775 & 0.81775 & 0.81775 & 0.81775 \\
& 500 & 0.82485 & 0.82515 & 0.82456 & 0.81997 & 0.82751 \\
& 1000 & 0.82219 & 0.82145 & 0.82115 & 0.81908 & 0.82751 \\
& 2000 & 0.81893 & 0.81716 & 0.81997 & 0.81953 & 0.82485 \\
& 4000 & 0.81834 & 0.81509 & 0.81538 & 0.8213 & 0.82352 \\
& 8000 & 0.8176 & 0.81538 & 0.81864 & 0.81967 & 0.8213 \\
& 16000 & 0.8216 & 0.81524 & 0.82249 & 0.82086 & 0.82263 \\
& 32000 & 0.82234 & 0.81065 & 0.82293 & 0.82396 & 0.8213 \\
\hline
\multirow{8}{*}{NDCG@20} 
& 0 & 0.51949 & 0.51949 & 0.51949 & 0.51949 & 0.51949 \\
& 500 & 0.51899 & 0.51739 & 0.52001 & 0.52728 & 0.52425 \\
& 1000 & 0.51736 & 0.51491 & 0.51661 & 0.52742 & 0.52532 \\
& 2000 & 0.51454 & 0.51056 & 0.51409 & 0.52841 & 0.52395 \\
& 4000 & 0.51111 & 0.50845 & 0.51221 & 0.53023 & 0.52333 \\
& 8000 & 0.5104 & 0.50777 & 0.51101 & 0.53049 & 0.51719 \\
& 16000 & 0.51354 & 0.50923 & 0.51423 & 0.53244 & 0.52286 \\
& 32000 & 0.51654 & 0.50718 & 0.51534 & 0.5389 & 0.52183 \\
\hline
\multirow{8}{*}{Recall@20} 
& 0 & 0.42887 & 0.42887 & 0.42887 & 0.42887 & 0.42887 \\
& 500 & 0.45515 & 0.4531 & 0.45329 & 0.44183 & 0.4552 \\
& 1000 & 0.45496 & 0.45392 & 0.45569 & 0.43807 & 0.45786 \\
& 2000 & 0.45523 & 0.45265 & 0.45518 & 0.44009 & 0.45636 \\
& 4000 & 0.45337 & 0.45062 & 0.45393 & 0.44354 & 0.45623 \\
& 8000 & 0.45316 & 0.45215 & 0.45542 & 0.44278 & 0.4527 \\
& 16000 & 0.45846 & 0.45417 & 0.46028 & 0.44532 & 0.4598 \\
& 32000 & 0.46303 & 0.45253 & 0.46247 & 0.44889 & 0.4607 \\
\hline
\multirow{8}{*}{Success@20} 
& 0 & 0.85932 & 0.85932 & 0.85932 & 0.85932 & 0.85932 \\
& 500 & 0.8645 & 0.86302 & 0.86479 & 0.86464 & 0.86686 \\
& 1000 & 0.86361 & 0.86302 & 0.86479 & 0.8608 & 0.86672 \\
& 2000 & 0.86154 & 0.86036 & 0.86346 & 0.86154 & 0.86627 \\
& 4000 & 0.8605 & 0.85784 & 0.86154 & 0.86183 & 0.86405 \\
& 8000 & 0.85962 & 0.85873 & 0.8605 & 0.86139 & 0.86065 \\
& 16000 & 0.86124 & 0.85976 & 0.86228 & 0.8605 & 0.86405 \\
& 32000 & 0.86686 & 0.85888 & 0.86391 & 0.86109 & 0.86391 \\
\hline
\multirow{8}{*}{NDCG@100} 
& 0 & 0.62251 & 0.62251 & 0.62251 & 0.62251 & 0.62251 \\
& 500 & 0.62607 & 0.62481 & 0.62697 & 0.62901 & 0.62989 \\
& 1000 & 0.62446 & 0.62222 & 0.62391 & 0.62902 & 0.63085 \\
& 2000 & 0.62246 & 0.61865 & 0.62162 & 0.62995 & 0.62963 \\
& 4000 & 0.61939 & 0.61735 & 0.62012 & 0.63217 & 0.62946 \\
& 8000 & 0.62001 & 0.61741 & 0.62006 & 0.63178 & 0.62407 \\
& 16000 & 0.62246 & 0.61867 & 0.62344 & 0.63372 & 0.62911 \\
& 32000 & 0.62528 & 0.61683 & 0.62496 & 0.63913 & 0.62899 \\
\hline
\multirow{8}{*}{Recall@100} 
& 0 & 0.61507 & 0.61507 & 0.61507 & 0.61507 & 0.61507 \\
& 500 & 0.65763 & 0.65614 & 0.65534 & 0.62793 & 0.65309 \\
& 1000 & 0.65757 & 0.65696 & 0.65847 & 0.62328 & 0.6562 \\
& 2000 & 0.65941 & 0.6572 & 0.6587 & 0.62496 & 0.65451 \\
& 4000 & 0.65889 & 0.65706 & 0.65834 & 0.63125 & 0.65589 \\
& 8000 & 0.66186 & 0.66088 & 0.66144 & 0.62854 & 0.65358 \\
& 16000 & 0.66628 & 0.66374 & 0.66767 & 0.63196 & 0.66051 \\
& 32000 & 0.67237 & 0.66268 & 0.6734 & 0.63364 & 0.66503 \\
\hline
\multirow{8}{*}{Success@100} 
& 0 & 0.91716 & 0.91716 & 0.91716 & 0.91716 & 0.91716 \\
& 500 & 0.92604 & 0.92559 & 0.92352 & 0.91967 & 0.92382 \\
& 1000 & 0.9247 & 0.92367 & 0.92411 & 0.9176 & 0.92485 \\
& 2000 & 0.92441 & 0.92278 & 0.92411 & 0.9182 & 0.92337 \\
& 4000 & 0.92367 & 0.92293 & 0.92293 & 0.92086 & 0.925 \\
& 8000 & 0.92574 & 0.92337 & 0.92456 & 0.91982 & 0.92367 \\
& 16000 & 0.92589 & 0.92322 & 0.92515 & 0.92012 & 0.9253 \\
& 32000 & 0.92885 & 0.92204 & 0.92618 & 0.92071 & 0.92633 \\
\hline
\end{tabular}
}
\caption{Full Results on TriviaQA.}
\label{tab:triviaqa_full}
\end{table}
\begin{table}
\centering
\resizebox{0.7\columnwidth}{!}{
\begin{tabular}{lccccccc}
\hline
Metric & \#iterations & FT & RP25 & InfoBatch & SP25 & DP \\
\hline
\multirow{8}{*}{NDCG@10} 
& 0 & 0.21935 & 0.21935 & 0.21935 & 0.21935 & 0.21935 \\
& 500 & 0.24666 & 0.24734 & 0.24895 & 0.24686 & 0.25367 \\
& 1000 & 0.25109 & 0.24906 & 0.25135 & 0.25058 & 0.25591 \\
& 2000 & 0.25992 & 0.26174 & 0.25768 & 0.25393 & 0.26853 \\
& 4000 & 0.28174 & 0.27356 & 0.27998 & 0.25811 & 0.28696 \\
& 8000 & 0.29324 & 0.27947 & 0.29348 & 0.26045 & 0.30409 \\
& 16000 & 0.30013 & 0.27998 & 0.30069 & 0.26549 & 0.30805 \\
& 32000 & 0.29779 & 0.27999 & 0.29524 & 0.26966 & 0.30926 \\
\hline
\multirow{8}{*}{Recall@10} 
& 0 & 0.11412 & 0.11412 & 0.11412 & 0.11412 & 0.11412 \\
& 500 & 0.13307 & 0.13346 & 0.13356 & 0.13029 & 0.13579 \\
& 1000 & 0.13501 & 0.13431 & 0.13513 & 0.13227 & 0.13827 \\
& 2000 & 0.13819 & 0.13903 & 0.13794 & 0.1331 & 0.14297 \\
& 4000 & 0.14543 & 0.14158 & 0.14237 & 0.13445 & 0.14976 \\
& 8000 & 0.15074 & 0.14488 & 0.1542 & 0.13591 & 0.15744 \\
& 16000 & 0.15776 & 0.1449 & 0.15936 & 0.1378 & 0.1625 \\
& 32000 & 0.15814 & 0.14411 & 0.15668 & 0.13863 & 0.16185 \\
\hline
\multirow{8}{*}{Success@10} 
& 0 & 0.63489 & 0.63489 & 0.63489 & 0.63489 & 0.63489 \\
& 500 & 0.70638 & 0.70383 & 0.70553 & 0.68681 & 0.71319 \\
& 1000 & 0.71489 & 0.71064 & 0.72255 & 0.69787 & 0.71915 \\
& 2000 & 0.73277 & 0.73702 & 0.73447 & 0.70043 & 0.74128 \\
& 4000 & 0.76255 & 0.75319 & 0.76255 & 0.71149 & 0.76766 \\
& 8000 & 0.78638 & 0.75915 & 0.79404 & 0.71489 & 0.79064 \\
& 16000 & 0.7966 & 0.75745 & 0.7966 & 0.72936 & 0.80085 \\
& 32000 & 0.78894 & 0.7583 & 0.78979 & 0.73617 & 0.8 \\
\hline
\multirow{8}{*}{NDCG@20} 
& 0 & 0.26348 & 0.26348 & 0.26348 & 0.26348 & 0.26348 \\
& 500 & 0.29366 & 0.29385 & 0.29566 & 0.29126 & 0.29826 \\
& 1000 & 0.29637 & 0.29634 & 0.29706 & 0.29559 & 0.3002 \\
& 2000 & 0.30455 & 0.30416 & 0.3036 & 0.30001 & 0.31348 \\
& 4000 & 0.326 & 0.31997 & 0.32444 & 0.30256 & 0.33175 \\
& 8000 & 0.33759 & 0.32329 & 0.33554 & 0.30537 & 0.34681 \\
& 16000 & 0.33797 & 0.32148 & 0.33982 & 0.31071 & 0.35008 \\
& 32000 & 0.33739 & 0.32248 & 0.33678 & 0.31654 & 0.35104 \\
\hline
\multirow{8}{*}{Recall@20} 
& 0 & 0.18949 & 0.18949 & 0.18949 & 0.18949 & 0.18949 \\
& 500 & 0.21663 & 0.21695 & 0.21657 & 0.20815 & 0.21731 \\
& 1000 & 0.21899 & 0.21913 & 0.22018 & 0.21162 & 0.21994 \\
& 2000 & 0.22352 & 0.22182 & 0.22289 & 0.21433 & 0.22791 \\
& 4000 & 0.22958 & 0.22733 & 0.22965 & 0.21376 & 0.23769 \\
& 8000 & 0.24058 & 0.22914 & 0.23801 & 0.21579 & 0.24658 \\
& 16000 & 0.24178 & 0.22593 & 0.24426 & 0.21874 & 0.25091 \\
& 32000 & 0.24346 & 0.22538 & 0.2442 & 0.22014 & 0.25207 \\
\hline
\multirow{8}{*}{Success@20} 
& 0 & 0.75319 & 0.75319 & 0.75319 & 0.75319 & 0.75319 \\
& 500 & 0.82128 & 0.81447 & 0.81872 & 0.79234 & 0.81617 \\
& 1000 & 0.81702 & 0.82298 & 0.82128 & 0.8 & 0.81532 \\
& 2000 & 0.82553 & 0.82043 & 0.82638 & 0.80426 & 0.84085 \\
& 4000 & 0.84426 & 0.85021 & 0.84426 & 0.8017 & 0.85702 \\
& 8000 & 0.86468 & 0.85362 & 0.86383 & 0.80936 & 0.87149 \\
& 16000 & 0.86553 & 0.8366 & 0.86553 & 0.81702 & 0.86979 \\
& 32000 & 0.86298 & 0.84 & 0.85957 & 0.81957 & 0.87064 \\
\hline
\multirow{8}{*}{NDCG@100} 
& 0 & 0.4582 & 0.4582 & 0.4582 & 0.4582 & 0.4582 \\
& 500 & 0.49334 & 0.49328 & 0.4953 & 0.48779 & 0.49772 \\
& 1000 & 0.49687 & 0.49633 & 0.49899 & 0.49213 & 0.50135 \\
& 2000 & 0.50644 & 0.5062 & 0.50519 & 0.49567 & 0.51531 \\
& 4000 & 0.52799 & 0.51896 & 0.52622 & 0.49817 & 0.53062 \\
& 8000 & 0.53567 & 0.52151 & 0.53459 & 0.50174 & 0.54524 \\
& 16000 & 0.53825 & 0.51991 & 0.53858 & 0.50578 & 0.5488 \\
& 32000 & 0.53734 & 0.51888 & 0.53583 & 0.50892 & 0.54931 \\
\hline
\multirow{8}{*}{Recall@100} 
& 0 & 0.44129 & 0.44129 & 0.44129 & 0.44129 & 0.44129 \\
& 500 & 0.48578 & 0.48488 & 0.48621 & 0.46503 & 0.48353 \\
& 1000 & 0.49008 & 0.48937 & 0.49335 & 0.46952 & 0.49181 \\
& 2000 & 0.49738 & 0.49507 & 0.49835 & 0.47295 & 0.50321 \\
& 4000 & 0.50844 & 0.50081 & 0.50844 & 0.47397 & 0.51132 \\
& 8000 & 0.51516 & 0.50016 & 0.51674 & 0.47688 & 0.52164 \\
& 16000 & 0.5223 & 0.49831 & 0.52344 & 0.47718 & 0.53231 \\
& 32000 & 0.52231 & 0.4936 & 0.52017 & 0.47535 & 0.53109 \\
\hline
\multirow{8}{*}{Success@100} 
& 0 & 0.91319 & 0.91319 & 0.91319 & 0.91319 & 0.91319 \\
& 500 & 0.93617 & 0.93532 & 0.93617 & 0.92511 & 0.93362 \\
& 1000 & 0.93702 & 0.93957 & 0.94298 & 0.92936 & 0.94128 \\
& 2000 & 0.94468 & 0.94979 & 0.94638 & 0.93191 & 0.95404 \\
& 4000 & 0.95149 & 0.95064 & 0.94979 & 0.93021 & 0.95404 \\
& 8000 & 0.95149 & 0.94809 & 0.94809 & 0.93532 & 0.95234 \\
& 16000 & 0.94894 & 0.94298 & 0.94979 & 0.93702 & 0.95574 \\
& 32000 & 0.94894 & 0.94213 & 0.94723 & 0.93787 & 0.95234 \\
\hline
\end{tabular}
}
\caption{Full Results on TripClick (head).}
\label{tab:tripclick_head_full}
\end{table}
\begin{table}
\centering
\resizebox{0.7\columnwidth}{!}{
\begin{tabular}{lccccccc}
\hline
Metric & \#iterations & FT & RP25 & InfoBatch & SP25 & DP \\
\hline
\multirow{8}{*}{NDCG@10} 
& 0 & 0.20834 & 0.20834 & 0.20834 & 0.20834 & 0.20834 \\
& 500 & 0.21923 & 0.21885 & 0.21742 & 0.2224 & 0.22026 \\
& 1000 & 0.21736 & 0.21728 & 0.21675 & 0.22483 & 0.22031 \\
& 2000 & 0.21753 & 0.21787 & 0.21744 & 0.22763 & 0.22117 \\
& 4000 & 0.22105 & 0.21873 & 0.22282 & 0.22932 & 0.22612 \\
& 8000 & 0.23183 & 0.22836 & 0.23209 & 0.22973 & 0.23327 \\
& 16000 & 0.24165 & 0.2365 & 0.24048 & 0.23194 & 0.24341 \\
& 32000 & 0.24507 & 0.23823 & 0.24798 & 0.23579 & 0.24942 \\
\hline
\multirow{8}{*}{Recall@10} 
& 0 & 0.21973 & 0.21973 & 0.21973 & 0.21973 & 0.21973 \\
& 500 & 0.23999 & 0.23883 & 0.23861 & 0.23966 & 0.24437 \\
& 1000 & 0.24527 & 0.24226 & 0.24388 & 0.24426 & 0.24363 \\
& 2000 & 0.2469 & 0.24785 & 0.2479 & 0.24509 & 0.2488 \\
& 4000 & 0.24662 & 0.24813 & 0.24984 & 0.24903 & 0.25319 \\
& 8000 & 0.25737 & 0.25676 & 0.25413 & 0.25152 & 0.25867 \\
& 16000 & 0.26549 & 0.26493 & 0.26218 & 0.25323 & 0.27092 \\
& 32000 & 0.27042 & 0.26894 & 0.27346 & 0.2574 & 0.27578 \\
\hline
\multirow{8}{*}{Success@10} 
& 0 & 0.52085 & 0.52085 & 0.52085 & 0.52085 & 0.52085 \\
& 500 & 0.54468 & 0.54383 & 0.54128 & 0.53787 & 0.54979 \\
& 1000 & 0.54809 & 0.54809 & 0.54468 & 0.54723 & 0.54383 \\
& 2000 & 0.55149 & 0.55234 & 0.55404 & 0.54809 & 0.55149 \\
& 4000 & 0.55234 & 0.56085 & 0.56511 & 0.54979 & 0.56426 \\
& 8000 & 0.58553 & 0.58383 & 0.58043 & 0.54979 & 0.57277 \\
& 16000 & 0.59574 & 0.58979 & 0.59149 & 0.55745 & 0.59489 \\
& 32000 & 0.59404 & 0.59149 & 0.59745 & 0.57277 & 0.61021 \\
\hline
\multirow{8}{*}{NDCG@20} 
& 0 & 0.26082 & 0.26082 & 0.26082 & 0.26082 & 0.26082 \\
& 500 & 0.27579 & 0.27316 & 0.27365 & 0.2762 & 0.27515 \\
& 1000 & 0.27235 & 0.27205 & 0.27208 & 0.27726 & 0.27659 \\
& 2000 & 0.27267 & 0.27239 & 0.27311 & 0.28023 & 0.27754 \\
& 4000 & 0.27849 & 0.2753 & 0.27812 & 0.28306 & 0.28408 \\
& 8000 & 0.28709 & 0.28263 & 0.28792 & 0.28343 & 0.29082 \\
& 16000 & 0.29727 & 0.28966 & 0.29574 & 0.28594 & 0.29924 \\
& 32000 & 0.30304 & 0.29397 & 0.30349 & 0.28986 & 0.30561 \\
\hline
\multirow{8}{*}{Recall@20} 
& 0 & 0.33434 & 0.33434 & 0.33434 & 0.33434 & 0.33434 \\
& 500 & 0.36105 & 0.35785 & 0.35804 & 0.35711 & 0.36212 \\
& 1000 & 0.36231 & 0.35973 & 0.3617 & 0.35593 & 0.36389 \\
& 2000 & 0.36339 & 0.36408 & 0.36735 & 0.35766 & 0.3686 \\
& 4000 & 0.37043 & 0.36911 & 0.3716 & 0.36253 & 0.37939 \\
& 8000 & 0.38016 & 0.37795 & 0.37853 & 0.36541 & 0.38278 \\
& 16000 & 0.38796 & 0.38256 & 0.38589 & 0.36692 & 0.39138 \\
& 32000 & 0.39753 & 0.39027 & 0.39605 & 0.37266 & 0.39998 \\
\hline
\multirow{8}{*}{Success@20} 
& 0 & 0.6383 & 0.6383 & 0.6383 & 0.6383 & 0.6383 \\
& 500 & 0.67915 & 0.66468 & 0.66809 & 0.67149 & 0.67574 \\
& 1000 & 0.67404 & 0.66553 & 0.67149 & 0.67319 & 0.6834 \\
& 2000 & 0.67745 & 0.6766 & 0.67915 & 0.67404 & 0.6834 \\
& 4000 & 0.69021 & 0.68681 & 0.69106 & 0.68 & 0.69872 \\
& 8000 & 0.69447 & 0.69702 & 0.69277 & 0.68851 & 0.70043 \\
& 16000 & 0.70468 & 0.70043 & 0.70128 & 0.69362 & 0.71064 \\
& 32000 & 0.7183 & 0.71319 & 0.71319 & 0.69191 & 0.72 \\
\hline
\multirow{8}{*}{NDCG@100} 
& 0 & 0.36876 & 0.36876 & 0.36876 & 0.36876 & 0.36876 \\
& 500 & 0.38618 & 0.38505 & 0.38477 & 0.3857 & 0.38506 \\
& 1000 & 0.38433 & 0.38485 & 0.38481 & 0.38712 & 0.38806 \\
& 2000 & 0.38651 & 0.38577 & 0.38586 & 0.39046 & 0.3905 \\
& 4000 & 0.39201 & 0.38943 & 0.39245 & 0.39179 & 0.39496 \\
& 8000 & 0.3995 & 0.39603 & 0.40126 & 0.39208 & 0.40289 \\
& 16000 & 0.41015 & 0.40399 & 0.41008 & 0.39419 & 0.4117 \\
& 32000 & 0.41662 & 0.40836 & 0.41801 & 0.39733 & 0.41867 \\
\hline
\multirow{8}{*}{Recall@100} 
& 0 & 0.60293 & 0.60293 & 0.60293 & 0.60293 & 0.60293 \\
& 500 & 0.64261 & 0.64128 & 0.64145 & 0.63196 & 0.64077 \\
& 1000 & 0.64984 & 0.64984 & 0.64989 & 0.63542 & 0.64843 \\
& 2000 & 0.66019 & 0.65785 & 0.65856 & 0.6394 & 0.66177 \\
& 4000 & 0.66304 & 0.66648 & 0.66787 & 0.6402 & 0.66628 \\
& 8000 & 0.67083 & 0.67221 & 0.6732 & 0.64336 & 0.67385 \\
& 16000 & 0.68717 & 0.68287 & 0.6885 & 0.64703 & 0.68699 \\
& 32000 & 0.69987 & 0.68954 & 0.70005 & 0.64903 & 0.69874 \\
\hline
\multirow{8}{*}{Success@100} 
& 0 & 0.83915 & 0.83915 & 0.83915 & 0.83915 & 0.83915 \\
& 500 & 0.85957 & 0.86128 & 0.85872 & 0.85957 & 0.85617 \\
& 1000 & 0.86383 & 0.86553 & 0.86638 & 0.86128 & 0.86213 \\
& 2000 & 0.87234 & 0.87234 & 0.87234 & 0.86213 & 0.87149 \\
& 4000 & 0.8783 & 0.88085 & 0.88085 & 0.86213 & 0.8766 \\
& 8000 & 0.88085 & 0.8834 & 0.8817 & 0.86298 & 0.88426 \\
& 16000 & 0.88766 & 0.88681 & 0.89021 & 0.86553 & 0.89021 \\
& 32000 & 0.89787 & 0.89362 & 0.89617 & 0.86894 & 0.89362 \\
\hline
\end{tabular}
}
\caption{Full Results on TripClick (torso).}
\label{tab:tripclick_torso_full}
\end{table}
\begin{table}
\centering
\resizebox{0.7\columnwidth}{!}{
\begin{tabular}{lccccccc}
\hline
Metric & \#iterations & FT & RP25 & InfoBatch & SP25 & DP \\
\hline
\multirow{8}{*}{NDCG@10} 
& 0 & 0.86791 & 0.86791 & 0.86791 & 0.86791 & 0.86791 \\
& 500 & 0.87421 & 0.85787 & 0.87533 & 0.89561 & 0.87832 \\
& 1000 & 0.87297 & 0.85656 & 0.87007 & 0.89834 & 0.88468 \\
& 2000 & 0.87564 & 0.86304 & 0.87542 & 0.90043 & 0.88611 \\
& 4000 & 0.88167 & 0.87094 & 0.88168 & 0.90521 & 0.89146 \\
& 8000 & 0.88552 & 0.87471 & 0.88762 & 0.90846 & 0.89684 \\
& 16000 & 0.8915 & 0.87811 & 0.89261 & 0.911 & 0.90002 \\
& 32000 & 0.892 & 0.87517 & 0.89319 & 0.91485 & 0.90247 \\
\hline
\multirow{8}{*}{Recall@10} 
& 0 & 0.93622 & 0.93622 & 0.93622 & 0.93622 & 0.93622 \\
& 500 & 0.94158 & 0.93736 & 0.94238 & 0.94385 & 0.94226 \\
& 1000 & 0.94305 & 0.93727 & 0.94163 & 0.94312 & 0.94523 \\
& 2000 & 0.94305 & 0.9397 & 0.9424 & 0.94194 & 0.94457 \\
& 4000 & 0.94484 & 0.9412 & 0.9444 & 0.9424 & 0.9459 \\
& 8000 & 0.94635 & 0.94284 & 0.94723 & 0.9434 & 0.94892 \\
& 16000 & 0.94806 & 0.94366 & 0.9481 & 0.94293 & 0.95087 \\
& 32000 & 0.9486 & 0.94342 & 0.94868 & 0.93992 & 0.94965 \\
\hline
\multirow{8}{*}{Success@10} 
& 0 & 0.97735 & 0.97735 & 0.97735 & 0.97735 & 0.97735 \\
& 500 & 0.98305 & 0.97705 & 0.9835 & 0.98755 & 0.98455 \\
& 1000 & 0.98305 & 0.9772 & 0.98185 & 0.988 & 0.9868 \\
& 2000 & 0.9829 & 0.979 & 0.9826 & 0.98785 & 0.98545 \\
& 4000 & 0.9847 & 0.9799 & 0.9838 & 0.98875 & 0.9865 \\
& 8000 & 0.985 & 0.9805 & 0.9859 & 0.98965 & 0.9883 \\
& 16000 & 0.98515 & 0.9802 & 0.9853 & 0.98935 & 0.9892 \\
& 32000 & 0.9838 & 0.9799 & 0.98455 & 0.988 & 0.98695 \\
\hline
\multirow{8}{*}{NDCG@20} 
& 0 & 0.8723 & 0.8723 & 0.8723 & 0.8723 & 0.8723 \\
& 500 & 0.87815 & 0.86211 & 0.87925 & 0.89912 & 0.88211 \\
& 1000 & 0.87689 & 0.86165 & 0.87416 & 0.90187 & 0.88805 \\
& 2000 & 0.88005 & 0.86781 & 0.87991 & 0.90409 & 0.89001 \\
& 4000 & 0.88609 & 0.87588 & 0.88641 & 0.90866 & 0.89527 \\
& 8000 & 0.88992 & 0.87945 & 0.89184 & 0.91165 & 0.90086 \\
& 16000 & 0.89561 & 0.88274 & 0.89676 & 0.91402 & 0.90385 \\
& 32000 & 0.89564 & 0.87916 & 0.89703 & 0.91796 & 0.90634 \\
\hline
\multirow{8}{*}{Recall@20} 
& 0 & 0.95048 & 0.95048 & 0.95048 & 0.95048 & 0.95048 \\
& 500 & 0.95453 & 0.95144 & 0.95523 & 0.95469 & 0.95459 \\
& 1000 & 0.95569 & 0.95407 & 0.95477 & 0.95399 & 0.95601 \\
& 2000 & 0.95721 & 0.95493 & 0.95659 & 0.95319 & 0.95708 \\
& 4000 & 0.95903 & 0.95695 & 0.9596 & 0.95303 & 0.95791 \\
& 8000 & 0.96049 & 0.95829 & 0.96069 & 0.95324 & 0.96156 \\
& 16000 & 0.96127 & 0.95902 & 0.96142 & 0.95232 & 0.96278 \\
& 32000 & 0.96026 & 0.95607 & 0.96086 & 0.94967 & 0.96195 \\
\hline
\multirow{8}{*}{Success@20} 
& 0 & 0.9859 & 0.9859 & 0.9859 & 0.9859 & 0.9859 \\
& 500 & 0.9901 & 0.98635 & 0.9904 & 0.9925 & 0.9904 \\
& 1000 & 0.9898 & 0.98755 & 0.9895 & 0.99265 & 0.99175 \\
& 2000 & 0.9907 & 0.98755 & 0.9898 & 0.9928 & 0.99175 \\
& 4000 & 0.99145 & 0.98845 & 0.99205 & 0.9934 & 0.99205 \\
& 8000 & 0.9925 & 0.9892 & 0.9922 & 0.9937 & 0.99355 \\
& 16000 & 0.9919 & 0.9895 & 0.99175 & 0.99355 & 0.99355 \\
& 32000 & 0.98965 & 0.98605 & 0.9901 & 0.99295 & 0.99235 \\
\hline
\multirow{8}{*}{NDCG@100} 
& 0 & 0.87636 & 0.87636 & 0.87636 & 0.87636 & 0.87636 \\
& 500 & 0.88261 & 0.86725 & 0.88371 & 0.90305 & 0.88639 \\
& 1000 & 0.88132 & 0.8665 & 0.87887 & 0.90575 & 0.89229 \\
& 2000 & 0.88435 & 0.87259 & 0.88431 & 0.90809 & 0.89427 \\
& 4000 & 0.89035 & 0.88029 & 0.89042 & 0.91256 & 0.89958 \\
& 8000 & 0.89399 & 0.88387 & 0.89579 & 0.91537 & 0.90469 \\
& 16000 & 0.89971 & 0.88716 & 0.90076 & 0.91775 & 0.90758 \\
& 32000 & 0.89989 & 0.88413 & 0.90125 & 0.92182 & 0.91042 \\
\hline
\multirow{8}{*}{Recall@100} 
& 0 & 0.96671 & 0.96671 & 0.96671 & 0.96671 & 0.96671 \\
& 500 & 0.97215 & 0.97231 & 0.97283 & 0.97047 & 0.9714 \\
& 1000 & 0.97325 & 0.97381 & 0.97341 & 0.96951 & 0.97267 \\
& 2000 & 0.97417 & 0.97459 & 0.97396 & 0.96964 & 0.97374 \\
& 4000 & 0.97595 & 0.97473 & 0.97535 & 0.96888 & 0.97511 \\
& 8000 & 0.97648 & 0.97609 & 0.97618 & 0.96812 & 0.97659 \\
& 16000 & 0.97727 & 0.97661 & 0.97698 & 0.96719 & 0.97709 \\
& 32000 & 0.97751 & 0.9762 & 0.97807 & 0.96518 & 0.97812 \\
\hline
\multirow{8}{*}{Success@100} 
& 0 & 0.99205 & 0.99205 & 0.99205 & 0.99205 & 0.99205 \\
& 500 & 0.99565 & 0.9943 & 0.9958 & 0.99625 & 0.99565 \\
& 1000 & 0.9955 & 0.9949 & 0.99565 & 0.99655 & 0.9958 \\
& 2000 & 0.99505 & 0.99475 & 0.99505 & 0.99745 & 0.9958 \\
& 4000 & 0.9955 & 0.99445 & 0.9949 & 0.99745 & 0.9964 \\
& 8000 & 0.9958 & 0.99505 & 0.9955 & 0.99685 & 0.99655 \\
& 16000 & 0.9958 & 0.9949 & 0.9958 & 0.9967 & 0.9964 \\
& 32000 & 0.9958 & 0.99505 & 0.9961 & 0.9967 & 0.99685 \\
\hline
\end{tabular}
}
\caption{Full Results on FEVER.}
\label{tab:fever_full}
\end{table}
\begin{table}
\centering
\resizebox{0.7\columnwidth}{!}{
\begin{tabular}{lccccccc}
\hline
Metric & \#iterations & FT & RP25 & InfoBatch & SP25 & DP \\
\hline
\multirow{8}{*}{NDCG@10} 
& 0 & 0.78999 & 0.78999 & 0.78999 & 0.78999 & 0.78999 \\
& 500 & 0.79112 & 0.78914 & 0.78982 & 0.79181 & 0.79469 \\
& 1000 & 0.79348 & 0.7896 & 0.79108 & 0.79162 & 0.79549 \\
& 2000 & 0.79611 & 0.79243 & 0.79308 & 0.79228 & 0.79771 \\
& 4000 & 0.79654 & 0.79531 & 0.79701 & 0.79543 & 0.80125 \\
& 8000 & 0.80141 & 0.79758 & 0.80054 & 0.79639 & 0.80721 \\
& 16000 & 0.80386 & 0.79894 & 0.80298 & 0.80032 & 0.81139 \\
& 32000 & 0.80585 & 0.80083 & 0.80647 & 0.80275 & 0.81208 \\
\hline
\multirow{8}{*}{Recall@10} 
& 0 & 0.76583 & 0.76583 & 0.76583 & 0.76583 & 0.76583 \\
& 500 & 0.77434 & 0.77286 & 0.77286 & 0.74477 & 0.77056 \\
& 1000 & 0.77576 & 0.77441 & 0.77596 & 0.73943 & 0.774 \\
& 2000 & 0.77981 & 0.77873 & 0.7788 & 0.74099 & 0.77731 \\
& 4000 & 0.78379 & 0.7815 & 0.78325 & 0.74105 & 0.78325 \\
& 8000 & 0.78899 & 0.78379 & 0.78778 & 0.74038 & 0.78771 \\
& 16000 & 0.79433 & 0.78656 & 0.79318 & 0.74207 & 0.79365 \\
& 32000 & 0.80088 & 0.78623 & 0.80068 & 0.74362 & 0.79953 \\
\hline
\multirow{8}{*}{Success@10} 
& 0 & 0.95179 & 0.95179 & 0.95179 & 0.95179 & 0.95179 \\
& 500 & 0.95935 & 0.95881 & 0.95895 & 0.95422 & 0.96003 \\
& 1000 & 0.96057 & 0.95935 & 0.95989 & 0.95409 & 0.96057 \\
& 2000 & 0.96259 & 0.96003 & 0.96084 & 0.95409 & 0.96138 \\
& 4000 & 0.96138 & 0.95868 & 0.9603 & 0.95395 & 0.96502 \\
& 8000 & 0.96111 & 0.95841 & 0.96016 & 0.95449 & 0.96597 \\
& 16000 & 0.96003 & 0.95787 & 0.96003 & 0.95463 & 0.96637 \\
& 32000 & 0.96124 & 0.9576 & 0.96084 & 0.95544 & 0.96529 \\
\hline
\multirow{8}{*}{NDCG@20} 
& 0 & 0.80366 & 0.80366 & 0.80366 & 0.80366 & 0.80366 \\
& 500 & 0.80525 & 0.80333 & 0.80421 & 0.80514 & 0.80885 \\
& 1000 & 0.808 & 0.80409 & 0.80541 & 0.80501 & 0.80945 \\
& 2000 & 0.8101 & 0.80683 & 0.80703 & 0.80565 & 0.81199 \\
& 4000 & 0.81002 & 0.80926 & 0.81084 & 0.80884 & 0.81438 \\
& 8000 & 0.81462 & 0.81148 & 0.81433 & 0.80939 & 0.82042 \\
& 16000 & 0.81756 & 0.81293 & 0.81682 & 0.81337 & 0.82424 \\
& 32000 & 0.81933 & 0.81502 & 0.81977 & 0.81546 & 0.82478 \\
\hline
\multirow{8}{*}{Recall@20} 
& 0 & 0.80743 & 0.80743 & 0.80743 & 0.80743 & 0.80743 \\
& 500 & 0.81695 & 0.8154 & 0.81614 & 0.78521 & 0.81303 \\
& 1000 & 0.81938 & 0.81783 & 0.81924 & 0.78008 & 0.81594 \\
& 2000 & 0.82228 & 0.82208 & 0.821 & 0.7819 & 0.82019 \\
& 4000 & 0.82458 & 0.82316 & 0.82492 & 0.78217 & 0.82357 \\
& 8000 & 0.82876 & 0.82606 & 0.82924 & 0.78015 & 0.82829 \\
& 16000 & 0.83565 & 0.82917 & 0.83477 & 0.78143 & 0.83268 \\
& 32000 & 0.84159 & 0.82937 & 0.84126 & 0.78204 & 0.83808 \\
\hline
\multirow{8}{*}{Success@20} 
& 0 & 0.96705 & 0.96705 & 0.96705 & 0.96705 & 0.96705 \\
& 500 & 0.97434 & 0.97407 & 0.97394 & 0.96772 & 0.97556 \\
& 1000 & 0.97596 & 0.97529 & 0.97434 & 0.96718 & 0.97556 \\
& 2000 & 0.97623 & 0.97515 & 0.97448 & 0.96678 & 0.97637 \\
& 4000 & 0.97556 & 0.97448 & 0.97556 & 0.96691 & 0.97623 \\
& 8000 & 0.97529 & 0.97218 & 0.97583 & 0.96691 & 0.97826 \\
& 16000 & 0.97637 & 0.97259 & 0.97569 & 0.9684 & 0.97893 \\
& 32000 & 0.97691 & 0.97205 & 0.97569 & 0.96907 & 0.97812 \\
\hline
\multirow{8}{*}{NDCG@100} 
& 0 & 0.82198 & 0.82198 & 0.82198 & 0.82198 & 0.82198 \\
& 500 & 0.82227 & 0.82079 & 0.82131 & 0.82394 & 0.82592 \\
& 1000 & 0.82477 & 0.8212 & 0.82234 & 0.82389 & 0.82651 \\
& 2000 & 0.8265 & 0.82336 & 0.82383 & 0.82444 & 0.82881 \\
& 4000 & 0.82664 & 0.82613 & 0.82728 & 0.82724 & 0.83099 \\
& 8000 & 0.83127 & 0.82852 & 0.83067 & 0.828 & 0.83651 \\
& 16000 & 0.8334 & 0.82963 & 0.83301 & 0.83107 & 0.84023 \\
& 32000 & 0.83481 & 0.83163 & 0.83551 & 0.8332 & 0.84056 \\
\hline
\multirow{8}{*}{Recall@100} 
& 0 & 0.88292 & 0.88292 & 0.88292 & 0.88292 & 0.88292 \\
& 500 & 0.88785 & 0.88805 & 0.88717 & 0.86313 & 0.88386 \\
& 1000 & 0.8892 & 0.88926 & 0.89001 & 0.8582 & 0.88724 \\
& 2000 & 0.89061 & 0.89129 & 0.89061 & 0.85976 & 0.89041 \\
& 4000 & 0.89318 & 0.89332 & 0.89352 & 0.85814 & 0.89271 \\
& 8000 & 0.89757 & 0.89608 & 0.89683 & 0.85658 & 0.89534 \\
& 16000 & 0.90095 & 0.89797 & 0.90149 & 0.85442 & 0.89878 \\
& 32000 & 0.90554 & 0.89764 & 0.90628 & 0.85571 & 0.90365 \\
\hline
\multirow{8}{*}{Success@100} 
& 0 & 0.98825 & 0.98825 & 0.98825 & 0.98825 & 0.98825 \\
& 500 & 0.99082 & 0.99095 & 0.99068 & 0.98825 & 0.99122 \\
& 1000 & 0.99149 & 0.99122 & 0.99055 & 0.98798 & 0.99109 \\
& 2000 & 0.99149 & 0.99109 & 0.99176 & 0.98785 & 0.9919 \\
& 4000 & 0.99163 & 0.99149 & 0.99176 & 0.98866 & 0.99203 \\
& 8000 & 0.9919 & 0.99217 & 0.99082 & 0.98866 & 0.99203 \\
& 16000 & 0.99203 & 0.9919 & 0.99203 & 0.98866 & 0.99298 \\
& 32000 & 0.99257 & 0.99109 & 0.99203 & 0.98879 & 0.9923 \\
\hline
\end{tabular}
}
\caption{Full Results on HotpotQA.}
\label{tab:hotpotqa_full}
\end{table}


\begin{table}[ht]
\centering
\caption{Qwen3-Embedding-0.6B on NFCorpus.}
\label{tab:qwen3_detail_nfcorpus}
\resizebox{0.95\columnwidth}{!}{
\begin{tabular}{@{}lccccc@{}}
\toprule
\textbf{Metric} & \textbf{Pretrained} & \textbf{FT} & \textbf{InfoBatch} & \textbf{SP} & \textbf{DP} \\
\midrule
MRR{@}1 & 0.458 & 0.502 & 0.511 & 0.533 & 0.517 \\
MRR{@}5 & 0.543 & 0.572 & 0.580 & 0.595 & 0.587 \\
MRR{@}10 & 0.550 & 0.581 & 0.587 & 0.605 & 0.595 \\
MRR{@}20 & 0.555 & 0.584 & 0.590 & 0.607 & 0.598 \\
MRR{@}50 & 0.556 & 0.586 & 0.591 & 0.609 & 0.600 \\
MRR{@}100 & 0.557 & 0.586 & 0.592 & 0.609 & 0.600 \\
\midrule
NDCG{@}5 & 0.440 & 0.479 & 0.480 & 0.492 & 0.486 \\
NDCG{@}10 & 0.441 & 0.479 & 0.478 & 0.487 & 0.479 \\
NDCG{@}20 & 0.461 & 0.489 & 0.484 & 0.489 & 0.492 \\
NDCG{@}50 & 0.517 & 0.539 & 0.540 & 0.545 & 0.541 \\
NDCG{@}100 & 0.577 & 0.617 & 0.617 & 0.615 & 0.618 \\
\midrule
Recall{@}1 & 0.061 & 0.065 & 0.068 & 0.067 & 0.065 \\
Recall{@}5 & 0.136 & 0.162 & 0.165 & 0.155 & 0.169 \\
Recall{@}10 & 0.171 & 0.236 & 0.237 & 0.214 & 0.234 \\
Recall{@}20 & 0.211 & 0.314 & 0.308 & 0.265 & 0.311 \\
Recall{@}50 & 0.267 & 0.416 & 0.417 & 0.357 & 0.408 \\
Recall{@}100 & 0.329 & 0.518 & 0.517 & 0.443 & 0.502 \\
\midrule
Success{@}1 & 0.458 & 0.502 & 0.511 & 0.533 & 0.517 \\
Success{@}5 & 0.672 & 0.690 & 0.697 & 0.693 & 0.697 \\
Success{@}10 & 0.721 & 0.755 & 0.749 & 0.768 & 0.762 \\
Success{@}20 & 0.793 & 0.796 & 0.796 & 0.793 & 0.805 \\
Success{@}50 & 0.836 & 0.842 & 0.845 & 0.851 & 0.864 \\
Success{@}100 & 0.889 & 0.889 & 0.885 & 0.885 & 0.892 \\
\bottomrule
\end{tabular}
}
\end{table}

\begin{table}[ht]
\centering
\caption{Qwen3-Embedding-0.6B on FiQA. The pretrained model outperforms all finetuned methods, likely due to high-quality financial data in its pretraining corpus.}
\label{tab:qwen3_detail_fiqa}
\resizebox{0.95\columnwidth}{!}{
\begin{tabular}{@{}lccccc@{}}
\toprule
\textbf{Metric} & \textbf{Pretrained} & \textbf{FT} & \textbf{InfoBatch} & \textbf{SP} & \textbf{DP} \\
\midrule
MRR{@}1 & 0.469 & 0.392 & 0.394 & 0.390 & 0.384 \\
MRR{@}5 & 0.543 & 0.465 & 0.466 & 0.464 & 0.464 \\
MRR{@}10 & 0.555 & 0.474 & 0.473 & 0.473 & 0.471 \\
MRR{@}20 & 0.559 & 0.479 & 0.478 & 0.477 & 0.476 \\
MRR{@}50 & 0.561 & 0.480 & 0.480 & 0.480 & 0.478 \\
MRR{@}100 & 0.561 & 0.481 & 0.481 & 0.480 & 0.479 \\
\midrule
NDCG{@}5 & 0.469 & 0.420 & 0.428 & 0.423 & 0.423 \\
NDCG{@}10 & 0.511 & 0.454 & 0.455 & 0.456 & 0.450 \\
NDCG{@}20 & 0.540 & 0.483 & 0.484 & 0.478 & 0.476 \\
NDCG{@}50 & 0.568 & 0.506 & 0.508 & 0.506 & 0.501 \\
NDCG{@}100 & 0.587 & 0.517 & 0.520 & 0.522 & 0.515 \\
\midrule
Recall{@}1 & 0.245 & 0.199 & 0.197 & 0.205 & 0.192 \\
Recall{@}5 & 0.453 & 0.377 & 0.384 & 0.377 & 0.377 \\
Recall{@}10 & 0.552 & 0.446 & 0.439 & 0.448 & 0.432 \\
Recall{@}20 & 0.633 & 0.510 & 0.508 & 0.504 & 0.495 \\
Recall{@}50 & 0.724 & 0.574 & 0.583 & 0.584 & 0.565 \\
Recall{@}100 & 0.798 & 0.612 & 0.623 & 0.634 & 0.611 \\
\midrule
Success{@}1 & 0.469 & 0.392 & 0.394 & 0.390 & 0.384 \\
Success{@}5 & 0.674 & 0.586 & 0.596 & 0.590 & 0.583 \\
Success{@}10 & 0.761 & 0.657 & 0.644 & 0.656 & 0.634 \\
Success{@}20 & 0.821 & 0.715 & 0.728 & 0.713 & 0.707 \\
Success{@}50 & 0.886 & 0.767 & 0.781 & 0.781 & 0.769 \\
Success{@}100 & 0.927 & 0.802 & 0.816 & 0.832 & 0.802 \\
\bottomrule
\end{tabular}
}
\end{table}

\begin{table}[ht]
\centering
\caption{Qwen3-Embedding-0.6B on ANTIQUE.}
\label{tab:qwen3_detail_antique}
\resizebox{0.95\columnwidth}{!}{
\begin{tabular}{@{}lccccc@{}}
\toprule
\textbf{Metric} & \textbf{Pretrained} & \textbf{FT} & \textbf{InfoBatch} & \textbf{SP} & \textbf{DP} \\
\midrule
MRR{@}1 & 0.615 & 0.580 & 0.625 & 0.640 & 0.615 \\
MRR{@}5 & 0.718 & 0.682 & 0.707 & 0.738 & 0.696 \\
MRR{@}10 & 0.724 & 0.689 & 0.716 & 0.746 & 0.703 \\
MRR{@}20 & 0.727 & 0.692 & 0.718 & 0.747 & 0.706 \\
MRR{@}50 & 0.728 & 0.693 & 0.719 & 0.748 & 0.707 \\
MRR{@}100 & 0.728 & 0.693 & 0.719 & 0.748 & 0.707 \\
\midrule
NDCG{@}5 & 0.519 & 0.496 & 0.502 & 0.546 & 0.510 \\
NDCG{@}10 & 0.518 & 0.496 & 0.506 & 0.540 & 0.520 \\
NDCG{@}20 & 0.570 & 0.546 & 0.549 & 0.582 & 0.567 \\
NDCG{@}50 & 0.646 & 0.617 & 0.624 & 0.658 & 0.636 \\
NDCG{@}100 & 0.690 & 0.663 & 0.665 & 0.707 & 0.674 \\
\midrule
Recall{@}1 & 0.072 & 0.071 & 0.069 & 0.074 & 0.074 \\
Recall{@}5 & 0.214 & 0.202 & 0.197 & 0.231 & 0.207 \\
Recall{@}10 & 0.304 & 0.270 & 0.264 & 0.317 & 0.284 \\
Recall{@}20 & 0.398 & 0.340 & 0.329 & 0.396 & 0.353 \\
Recall{@}50 & 0.510 & 0.444 & 0.431 & 0.514 & 0.450 \\
Recall{@}100 & 0.585 & 0.516 & 0.495 & 0.597 & 0.513 \\
\midrule
Success{@}1 & 0.615 & 0.580 & 0.625 & 0.640 & 0.615 \\
Success{@}5 & 0.855 & 0.835 & 0.840 & 0.885 & 0.820 \\
Success{@}10 & 0.905 & 0.890 & 0.905 & 0.935 & 0.870 \\
Success{@}20 & 0.940 & 0.930 & 0.930 & 0.950 & 0.915 \\
Success{@}50 & 0.980 & 0.955 & 0.950 & 0.975 & 0.950 \\
Success{@}100 & 0.980 & 0.965 & 0.955 & 0.980 & 0.955 \\
\bottomrule
\end{tabular}
}
\end{table}

\begin{table}[ht]
\centering
\caption{Qwen3-Embedding-0.6B on TriviaQA.}
\label{tab:qwen3_detail_triviaqa}
\resizebox{0.95\columnwidth}{!}{
\begin{tabular}{@{}lccccc@{}}
\toprule
\textbf{Metric} & \textbf{Pretrained} & \textbf{FT} & \textbf{InfoBatch} & \textbf{SP} & \textbf{DP} \\
\midrule
MRR{@}1 & 0.519 & 0.519 & 0.517 & 0.538 & 0.547 \\
MRR{@}5 & 0.602 & 0.615 & 0.611 & 0.620 & 0.634 \\
MRR{@}10 & 0.611 & 0.624 & 0.619 & 0.627 & 0.642 \\
MRR{@}20 & 0.614 & 0.627 & 0.622 & 0.630 & 0.645 \\
MRR{@}50 & 0.616 & 0.628 & 0.624 & 0.631 & 0.646 \\
MRR{@}100 & 0.616 & 0.628 & 0.624 & 0.631 & 0.647 \\
\midrule
NDCG{@}5 & 0.461 & 0.483 & 0.478 & 0.489 & 0.501 \\
NDCG{@}10 & 0.467 & 0.489 & 0.484 & 0.493 & 0.504 \\
NDCG{@}20 & 0.502 & 0.522 & 0.515 & 0.525 & 0.535 \\
NDCG{@}50 & 0.564 & 0.587 & 0.582 & 0.584 & 0.598 \\
NDCG{@}100 & 0.605 & 0.632 & 0.626 & 0.623 & 0.640 \\
\midrule
Recall{@}1 & 0.086 & 0.094 & 0.093 & 0.093 & 0.097 \\
Recall{@}5 & 0.236 & 0.267 & 0.263 & 0.253 & 0.271 \\
Recall{@}10 & 0.316 & 0.360 & 0.354 & 0.334 & 0.363 \\
Recall{@}20 & 0.403 & 0.462 & 0.453 & 0.421 & 0.461 \\
Recall{@}50 & 0.514 & 0.590 & 0.584 & 0.529 & 0.588 \\
Recall{@}100 & 0.591 & 0.677 & 0.669 & 0.603 & 0.669 \\
\midrule
Success{@}1 & 0.519 & 0.519 & 0.517 & 0.538 & 0.547 \\
Success{@}5 & 0.735 & 0.766 & 0.759 & 0.752 & 0.766 \\
Success{@}10 & 0.795 & 0.826 & 0.821 & 0.803 & 0.826 \\
Success{@}20 & 0.844 & 0.871 & 0.864 & 0.848 & 0.867 \\
Success{@}50 & 0.886 & 0.912 & 0.905 & 0.888 & 0.912 \\
Success{@}100 & 0.911 & 0.935 & 0.927 & 0.911 & 0.931 \\
\bottomrule
\end{tabular}
}
\end{table}

\end{document}